\documentclass[runningheads]{llncs}

\usepackage[dvipsnames]{xcolor}
\usepackage[utf8]{inputenc}
\usepackage{amsmath, amstext}
\usepackage{amsfonts,amssymb}
\usepackage{setspace}
\usepackage{algorithm}
\usepackage{multirow}
\usepackage{algpseudocode}
 \usepackage{graphicx}
 \usepackage{prftree}
 \usepackage{times}

\newcommand{\zero}{\textbf{0}}
\newcommand{\oomit}[1]{}

\newcommand{\removal}{\!\setminus\!}

\definecolor{myblack}{gray}{0}
\definecolor{myblack2}{gray}{0}

\begin{document}
\title{Session Types With Multiple Senders Single Receiver (report version)\thanks{Supported by NSFC under grant No. 61972385, 62032024, and 62192732.}}
%
%

\author{Zekun Ji\inst{1,2}  \and
Shuling Wang\inst{1}\thanks{Corresponding author} \and
Xiong Xu\inst{1}}
\authorrunning{Z. Ji et al.}
%
\institute{Skate Key Lab. of Comp. Sci., Institute of Software, Chinese Academy of Sciences\\
\and 
University of Chinese Academy of Sciences, Beijing, China, 
\email{\{jizk,wangsl,xux\}@ios.ac.cn}
 }

\maketitle              
\begin{abstract}
Message passing is a fundamental element in software development, ranging from concurrent and mobile computing to distributed services, but it suffers from communication errors such as deadlocks. Session types are a typing discipline for enforcing safe structured interactions between multiple participants. However, each typed interaction is restricted to having one fixed sender and one fixed receiver. 
In this paper, we extend session types with existential branching types, to handle a common interaction pattern with multiple senders and a single receiver in a synchronized setting, i.e. a receiver is available to receive messages from multiple senders, and which sender actually participates in the interaction cannot be determined till execution. We build the type system with existential branching types, which retain the important properties induced by standard session types: type safety, progress (i.e. deadlock-freedom), and fidelity. We further provide a novel communication type system to guarantee progress of dynamically interleaved multiparty sessions, by abandoning the strong restrictions of existing type systems. Finally, we encode Rust multi-thread primitives in the extended session types to show its expressivity, which can be considered as an attempt to check the deadlock-freedom of Rust multi-thread programs. 

\end{abstract}

\section{Introduction}


Distributed and concurrent programming plays an increasingly important role due to the high demand for distributed applications and services across networks. Message passing is one fundamental element in these areas and its correctness is very crucial. Many existing programming languages provide communication primitives, but still leave to the programmers the responsibility of guaranteeing  safety~\cite{DBLP:conf/asplos/TuLSZ19}. Concurrent and distributed programming suffers from communication errors such as deadlocks, and how to guarantee the correctness of communication behaviors is challenging.  

There have been many academic studies on the specification and verification of communicating behaviors. Session types~\cite{honda1993types,DBLP:conf/esop/HondaVK98} are a type theory for describing structured communications between multiple end-point participants to statically ensure safe interactions of them. It has been studied extensively in the context of process calculi~\cite{barwell2022generalised,bettini2008global,denielou2011dynamic,gheri2022design,toninho2017certifying,yoshida2010parameterised} and also in many programming languages~\cite{chen2022ferrite,jespersen2015session,kokke2019Rusty,lagaillardie2020implementing}. A variety of interaction patterns can be captured by session types, via sequencing, branching, recursion, and channel mobility, however, each interaction is typed with one fixed sender and one fixed receiver. In reality, a receiver may be available to receive messages from a set of senders, and which sender actually synchronizes with the receiver to complete the interaction is not determined till execution, e.g.  Rust multiple-producer, single-consumer channels. Existing session types are not expressive enough to handle such communication behaviors. 

This paper extends session types with an $\emph{existential branching}$ type that allows to specify interactions with multiple senders and one receiver, enforcing that at execution one among a set of senders synchronizes with the receiver each time. With the addition of the existential branching type, we need to first re-establish the whole type theory, which should retain the critical properties of session types including type safety, progress (deadlock-freedom), and session fidelity (type adherence). These properties guarantee that communications of all the end-point processes together realize the global interaction protocol specified by a   session global type. 
However, same as existing works based on global types~\cite{scalas2019less,barwell2022generalised},   the latter two properties put very strong restrictions on processes, e.g. each process plays only one role in the protocol. Other alternative approaches on guaranteeing deadlock-freedom loose these restrictions, but instead they must obey strong order structures for nested channels~\cite{bettini2008global,denielou2011dynamic}, or require heavy syntax annotations on the usage of channels~\cite{kobayashi2005type,kobayashi2006new}.
In our approach, we present a novel communication type system that records the execution orders between communication events over different channels, and checks the existence of dependence loops in the transition closure of all the orders. In particular, in order to deal with the dynamic mobility of channels, we define a unification process to transfer the channel to be moved and its communication history to the target receiver.

At the end, to show the expressivity of the extended session types, we encode concurrent primitives of Rust, including multiple-producer, single-consumer channels, mutex, and read-write locks, in the extended session calculus. As a result, a Rust multi-thread program can be encoded as a process that can be checked by the type systems presented in this paper. 
In summary, the contribution of this paper includes the following three aspects:
\vspace{-1.5mm}
\begin{itemize}
\item We extend session type with an existential branching type to specify the interactions that allow multiple senders and one receiver, and establish the type system to ensure type safety,  deadlock-freedom, and session fidelity;
\item We further propose a communication type system for checking deadlock-freedom of concurrent processes, by dropping strong restrictions of existing approaches;
   \item We encode concurrent primitives of Rust as processes in the extended session calculus, which shows the possibility of checking the communication behaviors of Rust using our approach in the future. 
\end{itemize} 
\vspace{-1mm}

 After presenting the related work, the paper is organized as follows: Sect.~\ref{sec:running} gives a motivating example to illustrate our approach; Sect.~\ref{sec:pi} presents the extended session calculus and session types with existential branching resp.; Sect.~\ref{sec:type} defines the new type system and proves the theorems corresponding to type safety, progress, and fidelity, and Sect.~\ref{sec:communication} proposes the communication type system for ensuring progress without restrictions. Sect.~\ref{sec:application} encodes Rust multi-thread primitives in the extended session calculus, and Sect.~\ref{sec:conclusion} concludes the paper and addresses future work.

\subsection{Related Work}

Session types were introduced by Honda~\cite{honda1993types,DBLP:conf/esop/HondaVK98} initially for describing and validating interactions between binary parties, and then extended to more generalized multiparties~\cite{bettini2008global,DBLP:conf/popl/HondaYC08}. Session types, either follow a top-down approach to require the communications between multiple participants conform to a global type via projections~\cite{bettini2008global,DBLP:conf/popl/HondaYC08,DBLP:journals/jlp/GhilezanJPSY19,yoshida2019very}, or check local session types of communicating systems directly, by requiring them to  be compatible~\cite{DBLP:conf/esop/HondaVK98}, satisfy behavioural predicates corresponding to safety, liveness, etc~\cite{scalas2019less}, or by rely-guarantee reasoning~\cite{DBLP:journals/jlp/ScalasY18}. 
Session types have been extended for various purposes. Dynamic multirole session type introduces a universal polling operator to describe that an arbitrary number of participants share common interaction behaviors and they can dynamically join and leave~\cite{denielou2011dynamic}. The more general parameterized multiparty session types integrate session types with indexed dependent types, to specify the number of participants and messages as parameters that can dynamically change~\cite{yoshida2010parameterised}. Both of them focus on specifying communication topologies with arbitrary number of participants and universally quantified interactions. Dependent types are integrated into session types to specify and verify the dependency and safety of data that are exchanged among multiple participants~\cite{toninho2017certifying,DBLP:conf/fossacs/ToninhoY18,hanwen2017}.   
There are also extensions of session types on specifying crash-stop failures~\cite{barwell2022generalised} and continuous time motion behaviors~\cite{DBLP:journals/pacmpl/MajumdarYZ20}.   However, all these works remain the interaction patterns of original session types unchanged, with no support to existential branching interactions considered in this paper.

The progress property, also called deadlock-freedom, has been studied a lot for channel-based communicating systems. 
The types of channels in~\cite{kobayashi2005type,kobayashi2006new} were decorated explicitly with capabilities/obligations to specify the relations on channel-wise usages, thus they need heavy syntax annotations.  The type systems for ensuring progress in session types were studied in~\cite{DBLP:conf/tgc/Dezani-CiancaglinidY07,bettini2008global}, avoiding any tagging of channels or names in syntax.  However, the typing rules of~\cite{DBLP:conf/tgc/Dezani-CiancaglinidY07,bettini2008global} put strong restrictions on nested channels in order to keep the correct orders between different parties, and moreover, the output is assumed to be asynchronous. Our communication type system was inspired by ~\cite{DBLP:conf/tgc/Dezani-CiancaglinidY07,bettini2008global}, but abandons the restrictions and moreover both input and output are synchronous. There is also work on verifying progress by checking multiparty compatibility of endpoint communications~\cite{imai2022kmclib,DBLP:conf/cav/LangeY19}, via model checking or other verification techniques.


\oomit{ 
\paragraph{Deadlock-Freedom}

For properties such as no deadlock, there are also several research works that give detection methods. Capabilities and obligations are a traditional way of thinking about concurrent programs when proving progress~\cite{kobayashi2005type,kobayashi2006new}. This approach considers formalizing the prerequisite environment for a communication event to occur and the impact on the environment; in order for a communication event to occur successfully on the receiver's side, the sender is obligated to make this send; and whether it can fulfill this obligation is determined by whether it has the capability to successfully complete this communication. By checking the capability level and the obligation level, a concurrent program can be checked to be deadlock-free.

Another tool, Kmclib, defines the compatibility of $k$ participants, written $k-MC$, and, for each process, detects deadlock-free qualities by checking its own $k-MC$ properties and those of the participants with which it interacts~\cite{imai2022kmclib}. Hybrid type theory can characterize interacting sub-protocols, and in this way gives ideas for combining multiple sub-protocols into a global protocol and thus detecting deadlocks~\cite{gheri2023hybrid}. However, the inference of these custom properties usually requires solving complex models of satisfiability problems, so automated checking can only be applied to smaller parallel programs.}

Session types have also been embedded in mainstream programming languages to guarantee the correctness of communication behaviors~\cite{jespersen2015session,fowler2019exceptional,kokke2019Rusty,lagaillardie2020implementing,chen2022ferrite}. For more details, readers are referred to a recent survey~\cite{DBLP:conf/asplos/TuLSZ19}. The commonality of them is to define session channels as enriched types to express the sequences of messages that are expected to occur along the channels, the ordering of messages, or the synchronization between inputs and outputs, to guarantee deadlock-freedom or other safety properties.  
However, none of these efforts focus on how the native parallel programs can be directly checked for communication correctness, and the strong-typed session channels may be not user-friendly to use and only generate code for fixed patterns.

\oomit{
\paragraph{Implementation in Programming Languages}

Much work has considered embodying Session Type in strongly-typed languages such as C and Rust, so that the powerful type checkers provided by strongly-typed language compilers can also check the session types generated by typing communication events. Using the Rust programming language, it is possible to define types such that it satisfies the properties of a Session Type and thus can be used as an implementation of a Session Type ~\cite{jespersen2015session}. Rust concurrent security verification based on Session Type and Good Variation~\cite{fowler2019exceptional} adds support for exceptions and canceled sessions~\cite{kokke2019Rusty}. Using the Scribble toolchain, protocol-compliant correct code can be generated from a readable protocol~\cite{lagaillardie2020implementing}. Ferrite based on $SIL\text{l}_S$ is also an attempt to embed session types in Rust~\cite{chen2022ferrite}. However, none of these efforts focus on how Rust's native parallel program code can be directly verified, and still generate Rust code from the global protocol, which is not user-friendly and can only generate code for a fixed pattern.}

\section{A Motivating Example}
\label{sec:running}

Fig.~\ref{fig:satellite_fig} shows a motivating example of a satellite control system, which includes three parts: the controller \textbf{Ctrl}, the ground control \textbf{Grd}, and the gyroscope \textbf{Gys}. The controller receives data from the gyroscope periodically, then calculates the control command based on the received data and sends it to the gyroscope to obey in the next period;  meanwhile, the ground control also needs to communicate with the controller to request current data or manually modify a certain control parameter of the controller. Therefore,  the controller should be always ready to handle the requests from both the ground and the gyroscope. However, the ground and the gyroscope cannot play one role due to their different responsibilities. Session types cannot specify the interactions of this example. 


   \oomit{ Here, on the one hand, the data queried by the ground control from the control system should be the data sent by the gyroscope; on the other hand, the control commands issued by the control system are determined by the parameters modified by the ground control and the data sent by the gyroscope together. Therefore, the control system needs to have the ability to communicate with both the ground and the gyroscope.}

    Fig.~\ref{fig:satellite_fig} shows how to specify this example using our approach. $G$ defines the global type for the interactions between the controller, the ground control and the gyroscope from a global perspective.  Let $p_r$ denote the ground control, $p_y$  the gyroscope, and  $p_c$ the controller. The controller  $p_c$ first receives a message from ground $p_r$ or gyroscope $p_y$ and then takes further actions according to the source of this message. If the message is from $p_r$, then $p_r$ continues to send a query request $inq$  and receive a response from $p_c$, or send a modification request $mod$ to modify the control parameter and further send the value to be changed. If the message is from $p_y$, then the controller receives  $data$ from $p_y$ and sends a new control command $comm$ back. Local types $s[p_c]$ and $s[A]$, where $A =   \color{myblack2}\{{\color{myblack}p}_i\}_{i\in \{r, y\}}$, are projections from $G$ to the three roles resp., to define the types of communications they participate in from their local view.    
    Especially, $s[p_c]$  defines the local type for  $p_c$, while $s[A]$ defines the local types for $p_r$ and $p_y$ as a whole, as it is not determined statically which among $p_r$ and $p_y$ interacts with $p_c$ at each time.  
\vspace{-2mm}

   
\begin{figure}[htbp]
	\centering
	\begin{minipage}{0.35\linewidth}
		\centering		\includegraphics[width=\linewidth]{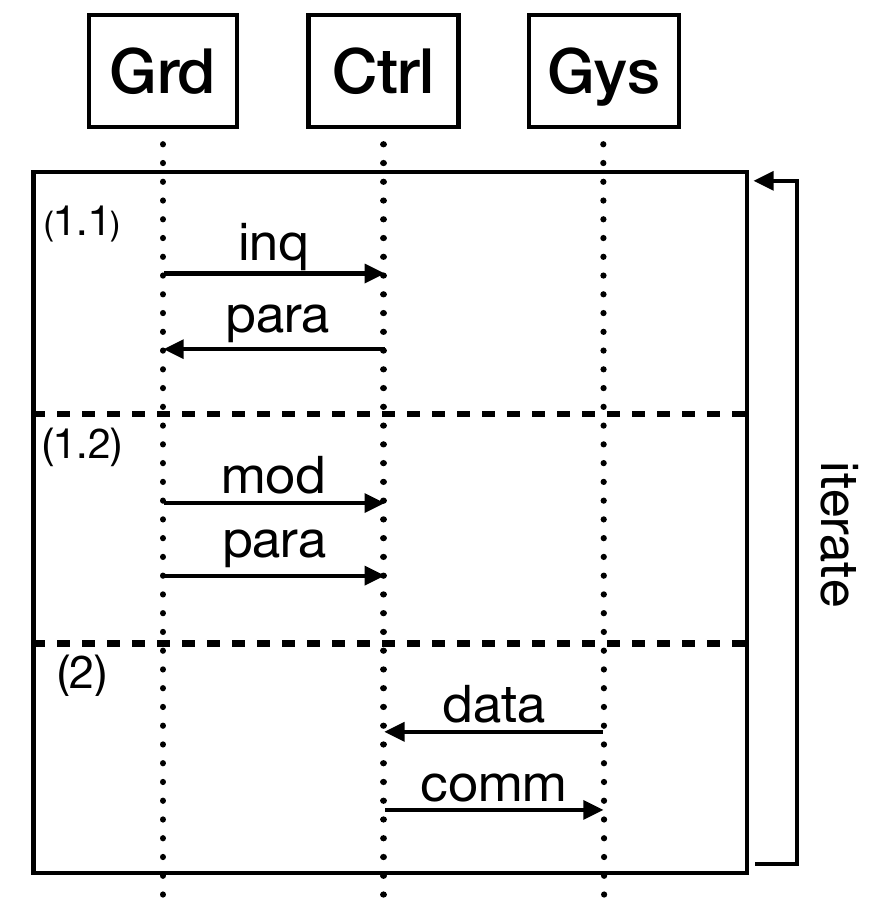}
		\label{fig:sate}
	\end{minipage}
	\begin{minipage}{0.64\linewidth}
		\centering
  \[
  \begin{scriptsize}
  \begin{array}{l}
  {\color{myblack}
            G {\color{black}=} \mu t. \exists_{i\in \{r, y\}} p_i \to p_c:}  \\
           \color{myblack}\quad \left\{\!
            \begin{array}{l}
                \text{l}_i(\_).p_i \!\to\! p_c\!:\!\left\{
                \begin{array}{l}
                     \text{inq}({\color{myblack2}\textbf{str}}). {{p_c}\!\to\! {p_i}}\!:\!\text{para}({\color{myblack2}\textbf{int}}). t \\
                     \text{mod}({\color{myblack2}\textbf{str}}). {{p_i}\!\to\! {p_c}}\!:\!\text{para}({\color{myblack2}\textbf{int}}). t
                \end{array}
                \right\}, i = r\\
                \text{l}_i(\_).p_i\! \to \!p_c\!:\!\left\{
                \begin{array}{l}
                \text{data}({\color{myblack2}\textbf{str}}). p_c\!\to\! p_i\!:\!\text{comm}({\color{myblack2}\textbf{int}}). t
                \end{array}
                \!\right\}, i = y
            \end{array}
            \right\}
  \end{array}
   \end{scriptsize}
  \] 
\vspace{-2mm}
   \[
   \begin{scriptsize}
   \begin{array}{lll}
    \color{myblack2}s[{\color{myblack}p}_c]:&\color{myblack2} \mu t. \exists_{i\in \{r, y\}} {\color{myblack}p}_i\&\\
    &\color{myblack2} \left\{
    \begin{array}{ll}
        \text{l}_i(\_).{\color{myblack}p}_i\&\left\{\!
        \begin{array}{l}
             \text{inq}({\color{myblack2}\textbf{str}}).{\color{myblack}p}_i\!\oplus\! \text{para}({\color{myblack2}\textbf{int}}). t \\
             \text{mod}({\color{myblack2}\textbf{str}}).{\color{myblack}p}_i\& \text{para}({\color{myblack2}\textbf{int}}). t
        \end{array}
        \!\right\}& , i = r\\
        \text{l}_i(\_).{\color{myblack}p}_i\&\{\text{data}({\color{myblack2}\textbf{str}}).{\color{myblack}p}_i\!\oplus\! \text{comm}({\color{myblack2}\textbf{int}}). t\} & ,i = y
    \end{array}
    \right\}\end{array}
    \end{scriptsize}
    \]
\vspace{-2mm}
    \[
    \begin{scriptsize}
    \begin{array}{lll}
    \color{myblack2}s[{\color{myblack}A}]:&\color{myblack2} \mu t.{\color{myblack}p}_c\oplus\\
    &\color{myblack2}\left\{
    \begin{array}{ll}
        \text{l}_i(\_).{\color{myblack}p}_c\!\oplus\! \left\{
        \begin{array}{l}
             \text{inq}({\color{myblack2}\textbf{str}}).{\color{myblack}p}_c\& \text{para}({\color{myblack2}\textbf{int}}). t \\
             \text{mod}({\color{myblack2}\textbf{str}}).{\color{myblack}p}_c\!\oplus\! \text{para}({\color{myblack2}\textbf{int}}). t
        \end{array}
        
        \right\}& , i = r\\
        \text{l}_i(\_).{\color{myblack}p}_c\!\oplus\! \{\text{data}({\color{myblack2}\textbf{str}}).{\color{myblack}p}_c\& \text{comm}({\color{myblack2}\textbf{int}}). t\} & ,i = y
    \end{array}
    \right\}\end{array}
    \end{scriptsize}
    \]
		\label{global_gyro}
	\end{minipage}
 \caption{Example of Attitude Control of Satellite}
 \label{fig:satellite_fig}
\end{figure}
\vspace{-5mm}

\section{MSSR Session Types}
\label{sec:pi}

This section introduces the extended session $\pi$-calculus, and the corresponding session types for typing the calculus.  We call the session calculus and types extended with multiple senders single receiver \emph{MSSR session calculus} and \emph{MSSR session types} resp. 

\subsection{Session $\pi$-calculus}

As an extension of~\cite{barwell2022generalised,scalas2019less}, MSSR session calculus is used to model processes interacting via multiparty channels. The syntax of it is given below, where $c$ stands for channels,  $d$ for data that can be transmitted during communication, including both basic values and channels, $P, Q$ for processes, and $D$ for function declarations. The sessions $s$ in restriction, and variables $x_i, y_i, \tilde{x}$ in branching and function declaration are bounded. We adopt the Barendregt convention that
bounded sessions and variables, and free ones are assumed pairwise distinct.
\oomit{The syntax of MSSR  $\pi$-calculus is given below, where $c$
denotes a channel, $x$ a variable, $s$ a session name, $p$ a participant (or a role), 
$d$ a value, $P,Q$ processes, $l, \text{l}_i$ labels, $D$ functions, $X$ process variables, and $I$  a non-empty indices.  }

\vspace{-2mm}
{\small 
\begin{equation*}
\begin{aligned}
c ::=\ & x\mid s[p] \quad &\mbox{variable, session with roles}\\ 
d ::=\ & x \mid a \mid s[p] & \mbox{variable, constant, session with roles}\\ 
P, Q ::=& \ \zero \mid (\nu s)P \mid P|Q & \mbox{inaction, restriction, parallel composition}\\
&\mid c[q]\oplus \text{l}\langle d\rangle.P & \mbox{selection towards $q$}\\
&\mid c[q]\& \{\text{l}_i(x_i).P_i\}_{i\in I} &\mbox{branching from $q$ with $I \neq \emptyset$}\\
&\mid \exists_{i \in I}c[q_i]\&\{\text{l}_i(y_i).P_i\} & \mbox{existential branching from $\{q_i\}_{i\in I}$}\\
&\mid  \textbf{def}\ D\ \textbf{in} \ P \mid X (\tilde d)& \mbox{process definition, call}\\ 
D ::=&\ X(\tilde x) = P&\mbox{function declaration}
\end{aligned}
\end{equation*}
}
Channels can be a variable, or $s[p]$ representing the channel of participant $p$ (called role $p$ interchangeably) in session $s$. Restriction $(\nu s)P$ declares a new session $s$ with the scope limited to process $P$. Parallel composition $P|Q$ executes $P$ and $Q$ concurrently, with communication actions between two roles synchronized. Selection $c[q]\oplus l\langle d\rangle.P$ sends value $d$ labelled by $l$ to role $q$, using channel $c$, and then continues as $P$. In contrary, branching  $c[q]\& \{\text{l}_i(x_i).P_i\}_{i\in I}$ expects to receive a value from role $q$ using channel $c$, and if a value $d_k$ with label $\text{l}_k$ is received for some $k\in I$, $x_k$ will be replaced by received value and the execution continues as the corresponding $P_k$. The $\text{l}_i$s must be pairwise distinct, and the scopes of $x_i$s are limited to $P_i$s. Existential branching $\exists_{i \in I}c[q_i]\&\{\text{l}_i(y_i).P_i\}$ is the newly added construct, which expects to receive a value from senders  $\{q_i\}_{i\in I}$, and if some $q_k$ sends a value,  $y_k$ will be replaced by the received value and the execution continues as $P_k$.  
Note the difference between branching and existential branching: the former waits for a message from a single role that may provide multiple messages, while the latter waits for a message from a set of different roles. 
For the above three kinds of prefix communications, we call $c$ their \emph{subject}. 

$X(\tilde x) = P$ declares the definition for process variable $X$, which can be called by  $X (\tilde d)$ with actual parameters $\tilde d$. Recursive processes can be modelled using it. 
$\textbf{def}\ D\ \allowbreak\textbf{in} \ P$ introduces the process definitions in $D$ and proceeds as $P$.   
To make sure that a recursive process has a unique solution, we assume $P$ in $X(\tilde x) = P$ is guarded~\cite{scalas2019less}, i.e. before the occurrence of $X$ in $P$, a prefix communication event occurs. For example, $P =  x_1[q] \oplus \text{l}\langle d \rangle. X(-)$ is guarded, while $P =  X(-)$ is not.

\oomit{
It should be pointed out that by combining the two branch operators, a more general scenario that expects to receive from a set of different roles providing a set of different labels can be realized, as shown below:
\[\exists_{i \in I}c[q_i]\&\{\text{l}_i(y_i).(c[q_i]\& \{\text{l}_{ij}(x_{ij}).P_{ij}\}_{j\in J})\}\]
}  
  
\paragraph{\textbf{Semantics}}

Before presenting the semantics, we define the 
reduction context $\mathbb{C}$ as follows: 
\[\mathbb{C} := \mathbb{C}|P \mid (\nu s) \mathbb{C} \mid \textbf{def}\ D\ \textbf{in}\ \mathbb{C} \mid [\ ] \] 
$\mathbb{C}$ is defined with a hole $[\ ]$ and for any process $P$, $\mathbb{C}[P]$ represents the process reached by substituting $P$ for $[\ ]$ in  $\mathbb{C}$. The execution of $\mathbb{C}[P]$ can always start with executing $P$. The semantics of the calculus are then defined by a set of reduction rules, denoted by $P \rightarrow P'$,  meaning that $P$ can execute to $P'$ in one step. We define $P \rightarrow^* P'$ to represent that $P $ executes to $P'$ via zero or more steps of reductions, and  $P\!\! \rightarrow$  to represent that there exists $P'$ such that $P \rightarrow P'$, and $P\!\!\nrightarrow$ otherwise. The semantics are given below:
\[	
{\small \begin{aligned}
			\relax[\&\oplus]&\quad s[p][q]\&\{\text{l}_i(x_i).P_i\}_{i\in I} | s[q][p]\oplus \text{l}_k\langle w \rangle .Q\ \to \ P_k[w/x_k] |Q & \mbox{if}\ k \in I\\
			[\exists\oplus]&\quad \exists_{i \in I}s[p][q_i]\&\{\text{l}_i(y_i).P_i\} | s[q_k][p]\oplus \text{l}_k\langle w\rangle .Q\ \to \ P_k[w/x_k]|Q & \mbox{if}\ k \in I
\\
\relax[X]&\quad \textbf{def}\ X(x_1, ..., x_n) = P\ \textbf{in}\ (X(w_1,...,w_n) |Q)\\ 
			& \to \textbf{def}\ X(x_1, ..., x_n) = P\ \textbf{in}\ (P[w_1/x_1,...,w_n/x_n]|Q)\\
			[Ctx]&\quad P \to P' \mbox{ implies } \mathbb{C}[P]\to \mathbb{C}[P']\\
			[\equiv]&\quad P'\equiv P\ \mbox{and}\ P\to Q\ \mbox{and}\ Q \equiv Q' \mbox{ implies } P'\to Q'\\
		\end{aligned}}
\]

Rule $[\&\oplus]$ defines the synchronized communication between roles $p$ (receiver) and $q$ (sender), matched by label $\text{l}_k$, resulting in the substitution of $w$ for  $x_k$  in continuation $P_k$, where $w$ is a constant $a$ or a session channel, say $s'[r]$. Rule $[\exists \oplus]$ defines the synchronized communication between receiver $p$ and a sender among $\{q_i\}_{i\in I}$, with the sender $q_k$ determined (externally from the sender side) for the communication. Rule $[X]$ replaces the occurrences of call $X(w_1, ..., w_n)$ by rolling its definition $P$ after substituting $w_i$s for corresponding parameters $x_i$s. Rule $[Ctx]$ defines that, if $P$ reduces to $P'$, then context $\mathbb{C}[P]$ reduces to $\mathbb{C}[P']$. Rule $[\equiv]$ defines that reduction is closed with respect to the structural equivalence. Here $P\equiv P'$ holds if $P$ can be reconstructed to $P'$ by $\alpha$-conversion for bounded variable renaming, or the commutativity and associativity laws of parallel composition operator.

\subsection{Global and Local Session Types}

Session types define the use of channels in $\pi$-calculus. A global type specifies the interactions of all the participants from a global view, and on the contrary a local type specifies the interactions of each participant from its local view. The projection of a global type upon a specific role produces the local type of the corresponding participant. Below we define the MSSR global and local types, denoted by $G$ and $T$ resp.
\vspace{-4mm}
 
\[ 
\small{
\begin{aligned}
	S ::= & \textbf{int}\mid \textbf{bool} \mid \textbf{real} \mid \langle G \rangle &
 B ::= & S \mid T  \\
	G :: = & p\rightarrow q: \{\text{l}_i(B_i).G_i\}_{i\in I} &
 T ::= & p \oplus \{\text{l}_i(B_i).T_i\}_{i\in I} \mid  p \&\{\text{l}_i(B_i).T_i\}_{i\in I}    \\ &\mid 
	        \exists_{i\in I} q_i \rightarrow p: \{\text{l}_i(B_i).G_i\} 
         &&\mid \exists_{i\in I}p_i\&\{\text{l}_i(B_i).T_i\}   \\ 
         & \mid \mu t. G \mid t \mid \textbf{end}
         &&\mid \mu t.T \mid  t \mid  \textbf{end} 
\end{aligned}}\]
 
\vspace{-2mm}
 Sort types $S$ for values include basic types and global types. 
  $B$ defines types of messages exchanged over channels, which can be sort types and channel types (denoted by $T$).  
 The global interaction type $p\rightarrow q: \{\text{l}_i(B_i).G_i\}_{i\in I}$ describes an interaction between role $p$ and $q$, saying that $p$ sends to $q$ one of the messages labelled by $\text{l}_i$, with payload type $B_i$ for some $i\in I$, and then continues according to the continuation type $G_i$. 
 The global existential interaction type $\exists_{i\in I} q_i \rightarrow p: \{\text{l}_i(B_i).G_i\}$ specifies that there exists $i\in I$ such that role $q_i$ sends to role $p$  a message, labelled by $\text{l}_i$ with payload  type $B_i$ and then continues according to $G_i$. We call $\{q_i\}_{i\in I}$ the \emph{domain} of the existential type. 
 $\mu t. G$ defines a recursive global type and $t$ represents a type variable. Type variables  $t$ are assumed to be guarded in the standard way. $\textbf{end}$ represents that no more communication will occur. 
 \oomit{, defined as follows:
 \[\textit{exdom(G)} \triangleq 
 \left \{
 \begin{array}{lll}
 	 \bigcup_{i \in I}\textit{exdom}(G_i)	& \mbox{if $G = p\rightarrow q: \{\text{l}_i(B_i).G_i\}_{i\in I}$}\\
 	\{\{q_i\}_{i \in I}\} \cup (\bigcup_{i \in I} \textit{exdom}(G_i))	& \mbox{if $G = \exists_{i\in I} q_i \rightarrow p: \{\text{l}_i(B_i).G_i\}$}\\
 	\textit{exdom}(G)	& \mbox{if $G = \mu t. G$}\\
 	\emptyset	& \mbox{otherwise}\\
 \end{array}
 \right. \]
 }

 Local types $T$ define the types of end-point channels.  
 The first three are corresponding to the selection, branching and existential branching processes in session $\pi$-calculus. The selection type (or internal choice) $p \oplus \{\text{l}_i(B_i).T_i\}_{i\in I}$ specifies a channel that acts as a sender and sends a message among $\{\text{l}_i(B_i)\}_{i \in I}$  to receiver $p$, and in contrary, the branching type (or external choice) $p \&\{\text{l}_i(B_i).T_i\}_{i\in I}$
  specifies a channel that acts as a receiver and expects to receive from sender $p$ a message among $\{\text{l}_i(B_i)\}_{i \in I}$, and for both of them, the corresponding continuations are specified as types $T_i$. The existential branching type (or existential external choice)
 $\exists_{i\in I}p_i\&\{\text{l}_i(B_i).T_i\}$  defines a receiving channel that expects to receive from a sender among $\{p_i\}_{i \in I}$ with corresponding message and continuation. 

 \begin{example} 
We use a simpler running example to explain the notations and definitions throughout this paper. 
The global type $G$ given below specifies a protocol between seller $r_s$, buyer $r_b$, and distributor $r_d$: 
$$
\small{
\begin{array}{ll}
     {\color{myblack} G} &\color{black}= \color{myblack}\exists_{i\in \{b, d\}} r_i \to r_s: \\
     &\color{myblack}\  \small{\left\{
     \begin{array}{ll}
        \text{purchase}.r_s\to r_i:\text{price}({\color{myblack2}\textbf{int}}).r_i\to r_s\{\text{ok}.\textbf{end}, \text{quit}.\textbf{end}\}&, i = b\\
        \text{deliver}.r_s\to r_i:\text{restock}({\color{myblack2}\textbf{str}}).\textbf{end} &, i = d
     \end{array}
     \right\}}
\end{array}
}
$$
At first, the seller expects to receive either a $\text{purchase}$ message from the buyer or a $\text{delivery}$ message from the distributor. If the former occurs, the seller sends the $\text{price}$ to the buyer, who in turn sends the seller its decision on whether or not to purchase; if the latter occurs,   the seller sends a message to the distributor about what it wants to $\text{restock}$.

Below presents the process, associated with $G$, for the three roles $r_s, r_b, r_d$ in the syntax of MSSR session calculus (we omit irrelevant message payloads). The buyer sends a purchase request to start a conversation with the seller, while the distributor not.    
$$
\small{
\begin{array}{l}
(\nu s)(P_{{\color{myblack}r}_s}|P_{{\color{myblack}r}_b}|P_{{\color{myblack}r}_d}),  
   \mbox{where}  \\
   \left\{
        \begin{array}{lll}
            &P_{{\color{myblack}r}_s}:& \exists_{i \in \{b,d\}} s[{\color{myblack}r}_s][{\color{myblack}r}_i]\& \small{\left\{
     \begin{array}{ll}
     \begin{array}{l}
        \text{purchase}.s[{\color{myblack}r}_s][{\color{myblack}r}_i]\oplus \text{price}(100)\\\qquad \qquad.s[{\color{myblack}r}_s][{\color{myblack}r}_i]\&\{\text{ok}.\textbf{0},\text{quit}.\textbf{0}\}\end{array}&, i = b\\
        \text{deliver}.s[{\color{myblack}r}_s][{\color{myblack}r}_i]\oplus \text{restock}\text{(``bread")}.\textbf{0} &, i = d
     \end{array}
     \right\}}\\
            &P_{{\color{myblack}r}_b}:& s[{\color{myblack}r}_b][{\color{myblack}r}_s]\oplus \text{purchase}.s[{\color{myblack}r}_b][{\color{myblack}r}_s]\&\text{price}(x).s[{\color{myblack}r}_b][{\color{myblack}r}_s]\oplus \text{ok}.\textbf{0}\\
            &P_{{\color{myblack}r}_d}:& \textbf{0}\\
        \end{array}
    \right.
\end{array}}
 $$ 
 
     \label{example:global}
 \end{example}





\subsection{Projection and Well-Formed Global Types}

The projection of a global type $G$ to a role $p$ is denoted by $G\upharpoonright_p$, which maps $G$ to the local type corresponding to $p$. 
Before defining $G\upharpoonright_p$, consider the projection of the existential branching type $\exists_{i\in I} q_i \rightarrow p: \{\text{l}_i(B_i).G_i\}$ to each $q_i$. At execution, there must exist one $i_0\in I$, such that $q_{i_0}$ sends to $p$ while others in $\{q_i\}_{i\neq i_0, i\in I}$ do not. The choice of $q_{i_0}$ is not determined statically, and meanwhile unique, i.e. there are not two $q_k, q_j$ with $k\neq j$ such that they both send to $p$. Due to this reason, we define the projection of the existential branching type to all $q_i$s as a whole. 
Instead of $G\upharpoonright_p$,   we define $G\upharpoonright_A$, where $A$ is a set of roles, that can be a singleton as before, or contains multiple roles corresponding to the domain of some existential type occurring in $G$. When $A$ is a singleton $\{r\}$, we write $G \upharpoonright_r$ for simplicity. $G\upharpoonright_A$ is defined as follows:
\[{\small
  \begin{array}{lll}
  	(p\rightarrow q: \{\text{l}_i(B_i).G_i\}_{i\in I})\upharpoonright_A \triangleq \left\{
  	\begin{array}{ll}
  		q\oplus\{\text{l}_i(B_i).G_i\upharpoonright_A\}_{i\in I} & \mbox{if $A=\{p\}$}\\
  		p\&\{\text{l}_i(B_i).G_i\upharpoonright_A\}_{i\in I} & \mbox{if $A=\{q\}$}\\
  		\sqcap\{G_i\upharpoonright_A\}_{i\in I} & \mbox{otherwise} 
  	\end{array}
  	\right.\\[2em] 
  	(\exists_{i\in I} q_i \rightarrow p: \{\text{l}_i(B_i).G_i\}) \upharpoonright_A \triangleq \left\{
  	\begin{array}{ll}
  		\exists_{i\in I}q_i\&\{\text{l}_i(B_i).G_i\upharpoonright_A\} & \mbox{if $A=\{p\}$}\\
  		p\oplus\{\text{l}_i(B_i).G_i\upharpoonright_{q_i}\}_{i\in I} & \mbox{if $A=\{q_i\}_{i\in I}$ }\\
  		\sqcap\{G_i\upharpoonright_A\}_{i\in I} & \mbox{otherwise} 
  	\end{array}
  	\right.\\
  	(\mu t. G)\upharpoonright_A \triangleq
  	\left \{
  	\begin{array}{ll}
  		\mu t. G\upharpoonright_A & \mbox{if 
  		$G\upharpoonright_A \neq t'$ for all $t'$}\\
  	\textbf{end} & \mbox{otherwise}
  	\end{array}
   	\right. 
   	\qquad t \upharpoonright_A = t \qquad 
   	\textbf{end} \upharpoonright_A = \textbf{end} 
  \end{array}}
\]
For global type $p\rightarrow q: \{\text{l}_i(B_i).G_i\}_{i\in I}$, the projection to sender $p$ or receiver $q$ results in the corresponding selection and branching local types; while the projection to $A$ that is neither $\{p\}$ nor $\{q\}$, is defined to be the merge of the projections of continuation types $G_i$ (denoted by $\sqcap$, to be explained below).  The third case indicates that, from the eye of any role rather than $p$ and $q$, it is not known which case among $I$ is chosen thus all the continuations in $I$ need to be merged to include all possible choices. 
The projection of $\exists_{i\in I} q_i \rightarrow p: \{\text{l}_i(B_i).G_i\}$ to receiver $p$ results in the existential branching type on the set of senders, while the projection to the set of senders $\{q_i\}_{i \in I}$ is defined as the selection type on receiver $p$. The projection of $\mu t. G$ to $r$ is defined inductively.  

\oomit{
As shown in above definitions, given a global type $G$ and a role $p$ of $G$, the projection of $G$ to $p$ contains two parts: $G\upharpoonright_p$ and $G\upharpoonright_{A_p}$ for any $A_p \in \textit{exdom}(G)$ satisfying $p \in A_p$. According to the syntax structure of $G$, the communications defined by $G\upharpoonright_p$ must occur in $G$ before  $G\upharpoonright_{A_p}$: as soon as $p$ enters an existence domain, the subsequent occurrence of $p$ is bound to this domain. 
}

Now we define the merge operator $\sqcap$ for two local types, presented below:
\vspace{-1mm}
\[
{\small
\begin{array}{l}
   \!\! \begin{array}{lll}
	&&p\&\{\text{l}_i(B_i).T_i\}_{i \in I}
	\sqcap 	p\&\{\text{l}_j(B_j).S_j\}_{j \in J} \\
	& =& 
            p\&\{\text{l}_k(B_k).(T_k \sqcap S_k)\}_{k \in I \cap J}\ \& \\
        &&    p\&\{\text{l}_i(B_i).T_i\}_{i \in I\setminus J}\ \& \\
        &&    p \&\{\text{l}_j(B_j).S_j\}_{j \in J\setminus I}
        \\[.4em]
    \end{array}
    \begin{array}{lll}
    &&	\exists_{i\in I}p_i\&\{\text{l}_i(B_i).T_i\} 
    		\sqcap 	\exists_{j\in J}p_j\&\{\text{l}_j(B_j).S_j\} \\
    		& = &
      \exists_{k\in I\cap J}p_k\&\{\text{l}_k(B_k).(T_k \sqcap S_k)\}\ \& \\
      &&\exists_{i\in I\setminus J}p_i\&\{\text{l}_i(B_i).T_i\}\ \& \\
      &&\exists_{j\in J\setminus I}p_j \&\{\text{l}_j(B_j).S_j\}
    \\[.4em]
    \end{array}\\[.4em]
	
		p\oplus\{\text{l}_i(B_i).T_i\}_{i \in I}
		\sqcap p\oplus\{\text{l}_i(B_i).T_i\}_{i \in I} = p\oplus\{\text{l}_i(B_i).T_i\}_{i \in I}\\[.4em]
\oomit{
&&p\&\{l(B).T\} \sqcap \exists_{j\in J}p_j\&\{\text{l}_j(B_j).S_j\} \\
	&=&  \left\{
 \begin{array}{llll}
 	\exists_{j\in J}p_j\&\{\text{l}_j(B_j).S_j\}\&
 	p\&\{l(B).T\}  &
 	\mbox{if $p\neq p_j$ for all $j \in J$}\\
 \textit{undefined} \quad & \mbox{ otherwise}	\\ 	
 \end{array}
 \right.  }
 	\mu t.T \sqcap \mu t. S = \mu t. (T\sqcap S) \quad t \sqcap t = t \quad \textbf{end}
	\sqcap \textbf{end} = \textbf{end}
\end{array}
}
\] 
The merge of two branching types $\&$ combines together to form a larger branching type to allow to receive more messages. However, two selection types $\oplus$ can only be merged when they have exactly the same choices. Otherwise, a selection type with more (internal) choices will produce unexpected behaviors than the global type as specified.

 
A global type $G$ is \emph{well-formed}, if $G \upharpoonright_A$ is defined
	for each $A \in \textit{roles}(G) \cup \textit{exdom}(G)$, where $\textit{roles}(G)$ returns the set of roles occurring in $G$ and $\textit{exdom}(G)$ returns the set of domains of all existential branching types occurring in $G$. From now on, we only consider well-formed global types.

\begin{example}
The projection of $G$ presented in Exam.~\ref{example:global} upon each role returns the corresponding local type.  Let $R = \{{\color{myblack}r}_i\}_{i\in \{b, d\}}$, then: 
$$
\color{myblack2}s[{\color{myblack}r}_s]: G\upharpoonright_{r_s} = \exists_{i\in \{b, d\}} {\color{myblack}r}_i\& 
      \small{\left\{
     \begin{array}{ll}
        \text{purchase}.{\color{myblack}r}_i\oplus \text{price}({\color{myblack2}\textbf{int}}).{\color{myblack}r}_i\&\{\text{ok}.\textbf{end}, \text{quit}.\textbf{end}\}&, i = b\\
        \text{deliver}.{\color{myblack}r}_i\oplus \text{restock}({\color{myblack2}\textbf{str}}).\textbf{end} &, i = d
     \end{array}
     \right\}}
$$
\vspace{-2mm}
$$
\color{myblack2}s[R]: G\upharpoonright_R = {\color{myblack}r}_i\oplus 
      \small{\left\{
     \begin{array}{ll}
        \text{purchase}.{\color{myblack}r}_s\& \text{price}({\color{myblack2}\textbf{int}}).{\color{myblack}r}_s\oplus \{\text{ok}, \text{quit}\}.\textbf{end}&, i = b\\
        \text{deliver}.{\color{myblack}r}_s\& \text{restock}({\color{myblack2}\textbf{str}}).\textbf{end} &, i = d
     \end{array}
     \right\}}
$$

\label{example:local}
\end{example}





\subsection{Consistency of Global Types and Local Types}
 
The transition semantics of types abstract the communications that might occur over typed channels. They can also be used to investigate the consistency between a global type and its projection. 
First, we define the typing context for channels:
\begin{definition}[Typing context for channels]
    The channel typing context $\Gamma$ maps a channel (e.g. $x, s[p]$) or a set of channels (e.g. $s[A]$) to local types, defined as:
	\[
	{\small \Gamma \triangleq \emptyset \mid \Gamma, s[p]: T \mid \Gamma, s[A]: T \mid \Gamma, x:T}	
	\]
 \label{def:contextchannel}
\end{definition}
\vspace{-5mm}
The \emph{domain} of each $\Gamma$ denotes the set of channels of the form $s[p]$ or $s[A]$. 
We write $\Gamma_1, \Gamma_2$ to denote the union of $\Gamma_1$ and $\Gamma_2$, when they have disjoint domains.


The transition semantics for global and local type is defined by transition relations $G \xrightarrow{\alpha}_G G'$ and $\Gamma \xrightarrow{\alpha'}_L \Gamma'$ resp. Here $\alpha$ represents a communication, while $\alpha'$ can be an input, output, or communication, as shown in the following rules. The typing transition semantics for global and local types are presented in Fig.~\ref{fig:fulltransitionglobal} and Fig.~\ref{fig:fulltransitionlocal} respectively. 

\begin{figure}
    \centering
 \begin{align}
		p\rightarrow q: \{\text{l}_i(B_i).G_i\}_{i\in I}
		\xrightarrow{p \rightarrow q: \text{l}_k(B_k)}  G_k \quad \forall k \in I \tag{G-comm}\\
  \frac{\forall i \in I. G_i \xrightarrow{\alpha} G_i' \quad \{p, q\} \cap \textit{roles}(\alpha) = \emptyset  }{p\rightarrow q: \{\text{l}_i(B_i).G_i\}_{i\in I}
		\xrightarrow{\alpha}  p\rightarrow q: \{\text{l}_i(B_i).G_i'\}_{i\in I}} \tag{G-comm'}
  \\
		\exists_{i\in I} q_i \rightarrow p: \{\text{l}_i(B_i).G_i\} 
		\xrightarrow{q_k \rightarrow p: \text{l}_k(B_k)}  G_k \quad \forall k \in I
		\tag{G-exist}\\
  \frac{\forall i \in I. G_i \xrightarrow{\alpha} G_i' \quad \{p, q_i\}_{i\in I} \cap \textit{roles}(\alpha) = \emptyset  }{\exists_{i\in I} q_i \rightarrow p: \{\text{l}_i(B_i).G_i\} 
		\xrightarrow{\alpha}  \exists_{i\in I} q_i \rightarrow p: \{\text{l}_i(B_i).G_i'}
  \tag{G-exist'}\\
		 \frac{G[\mu t. G/t] \xrightarrow{\alpha} G'}{\mu t. G \xrightarrow{\alpha} G'} \tag{G-rec}
	\end{align}
     \caption{Typing Semantics of MSSR Global Types}    \label{fig:fulltransitionglobal}
\end{figure}

\begin{figure}
    \centering
    \begin{align}
	 s[p]:q\oplus \{\text{l}_i(B_i).T_i\}_{i\in I} \xrightarrow{p:q\oplus \text{l}_k(B_k)}_L s[p]:T_k \mbox{  $\forall k\in I$}
	 \tag{L-select}\\ 
	s[p]:q\&\{\text{l}_i(B_i).T_i\}_{i\in I} \xrightarrow{p:q\&\text{l}_k(B_k)}_L s[p]:T_k 
	\mbox{  $\forall k\in I$} \tag{L-branch}\\
 s[p]:\exists_{i\in I}q_i\&\{\text{l}_i(B_i).T_i\} \xrightarrow{p:q_k\&\text{l}_k(B_k)}_L s[p]:T_k \mbox{  $\forall  k\in I$}\tag{L-exist}\\ 
  s[\{p_i\}_{i \in I}]:q\oplus \{\text{l}_i(B_i).T_i\}_{i\in I} \xrightarrow{p_k:q\oplus \text{l}_k(B_k)}_L s[p_k]:T_k \mbox{  $\forall  k\in I$}
 \tag{L-select'}\\ 
		\frac{\Gamma_1 \xrightarrow{p:q\oplus l(B_1)}_L \Gamma_1 '\quad \Gamma_2 \xrightarrow{q:p\&l(B_2)}_L \Gamma_2 '}
		{\Gamma_1, \Gamma_2 \xrightarrow{p \rightarrow q: l(B_1)}_L \Gamma_1', \Gamma_2'}\tag{L-par} \\
		\frac{T[\mu t. T/t] \xrightarrow{\alpha}_L T'} {\mu t. T \xrightarrow{\alpha}_L T'} \tag{L-rec}
	\end{align}
    \caption{Typing Semantics of MSSR Local Types}
\label{fig:fulltransitionlocal}
\end{figure}

\vspace{-1mm}
{\small
	\begin{align}
			\exists_{i\in I} q_i \rightarrow p: \{\text{l}_i(B_i).G_i\} 
		\xrightarrow{q_k \rightarrow p: \text{l}_k(B_k)}_G  G_k \quad \forall k \in I
		\tag{G-exist}\\
  \frac{\forall i \in I. G_i \xrightarrow{\alpha}_G G_i' \quad \{p, q_i\}_{i\in I} \cap \textit{roles}(\alpha) = \emptyset  }{\exists_{i\in I} q_i \rightarrow p: \{\text{l}_i(B_i).G_i\} 
		\xrightarrow{\alpha}_G  \exists_{i\in I} q_i \rightarrow p: \{\text{l}_i(B_i).G_i'\}}
  \tag{G-exist'}
  \\
  s[p]:\exists_{i\in I}q_i\&\{\text{l}_i(B_i).T_i\} \xrightarrow{p:q_k\&\text{l}_k(B_k)}_L s[p]:T_k \mbox{  $\forall  k\in I$}\tag{L-exist}\\ 
  s[\{p_i\}_{i \in I}]:q\oplus \{\text{l}_i(B_i).T_i\}_{i\in I} \xrightarrow{p_k:q\oplus \text{l}_k(B_k)}_L s[p_k]:T_k \mbox{  $\forall  k\in I$}
 \tag{L-select'}\\
 	\frac{\Gamma_1 \xrightarrow{p:q\oplus \text{l}(B)}_L \Gamma_1 '\quad \Gamma_2 \xrightarrow{q:p\&\text{l}(B)}_L \Gamma_2}
		{\Gamma_1, \Gamma_2 \xrightarrow{p \rightarrow q: \text{l}(B)}_L \Gamma_1', \Gamma_2'}\tag{L-par}
	\end{align}
 }

\noindent Rule (G-exist) performs a global interaction from a sender among $\{q_i\}_{i\in I}$ to $p$, while 
rule (G-exist')  defines the case which permutes the order of two actions that are causally unrelated. The action to be executed for (G-exist') is required to be available no matter what choice is made for the prefix interaction. For local types, rules (L-exist) and (L-select') define the input and output resp., and rule (L-par) defines the synchronization between an input and a compatible output.  
  
\paragraph{\textbf{Consistency of Projection}}

Suppose $G$ is a global type with roles $\{p_1, ..., p_k\}$ and domain sets $\{A_1, ..., A_l\}$. Let $\{s[p_1]: T_1, ..., s[p_k]: T_{k}; s[A_1]:T'_{1}, ..., s[A_l]:T'_{l}\}$ be the local types obtained from the projection of $G$ towards the roles, denoted by $\textit{proj}(G)$ below.  We give the definition of \emph{consistent} global type and local types below. 
\begin{definition}[Consistency]
	A global type $G$ and a set of local types $\Gamma$ are \emph{consistent}, if the following conditions hold: if $G \xrightarrow{p \rightarrow q: \text{l}(B)} G'$, then there must
		exist $\Gamma_1, \Gamma_2$ in $\Gamma$ with  $\Gamma = \Gamma_1, \Gamma_2, \Gamma_3$, such that
		$\Gamma_1, \Gamma_2 \xrightarrow{p \rightarrow q: \text{l}(B)} \Gamma_1', \Gamma_2'$, and let $\Gamma' = \Gamma_1', \Gamma_2', \Gamma_3$,  
		 $G'$ and $\Gamma' $  are consistent; and vice versa. 
\label{definition:consistency}
\end{definition}

\begin{theorem}
	Let $G$ be a well-formed global type, then $G$ and its projection $\textit{proj}(G)$ are consistent.
 \label{theorem:consistency}
\end{theorem} 

\begin{proof}
    We first prove that if $G\xrightarrow{\alpha} G'$, $\Gamma = proj(G)$, then there must exist $\Gamma_1, \Gamma_2$ in $\Gamma$ with  $\Gamma = \Gamma_1, \Gamma_2, \Gamma'$, such that $\Gamma_1, \Gamma_2 \xrightarrow{\alpha} \Gamma_1', \Gamma_2'$, and $proj(G') = \Gamma_1', \Gamma_2', \Gamma'$.

    We prove this fact by structural induction on the reduction rules for transition semantics of global types.
    
\noindent \textbf{Case G-comm}: We can safely assume that $G = p\rightarrow q: \{\text{l}_i(B_i).G_i\}_{i\in I}$, $\alpha = p \rightarrow q: \text{l}_k(B_k)$. By projection rule,$proj(G) = \Gamma_1, \Gamma_2, \Gamma'$, for $\Gamma_1 = s[p]: q \oplus \{\text{l}_i(B_i).G_i\upharpoonright_p\}_{i\in I}$, $\Gamma_2 = s[q]: p \& \{\text{l}_i(B_i).G_i\upharpoonright_p\}_{i\in I}$, and $s[p] \notin dom(\Gamma'), s[q]\notin dom(\Gamma')$. Then, by (L-select) and (L-branch), we can get $\Gamma_1 \xrightarrow{p:q\oplus \text{l}_k(B_k)} G_k\upharpoonright_p$ and $\Gamma_2 \xrightarrow{q:p\& \text{l}_k(B_k)} G_k\upharpoonright_q$. By (L-par) we have $\Gamma_1, \Gamma_2 \xrightarrow{p \rightarrow q: \text{l}_k(B_k)} G_k\upharpoonright_p, G_k\upharpoonright_q$. Since $s[p] \notin dom(\Gamma'), s[q]\notin dom(\Gamma')$, $proj(G') = G_k\upharpoonright_p, G_k\upharpoonright_q, \Gamma'$, so we have done.

\noindent \textbf{Case G-exist}: We assume that $G = \exists_{i\in I} p_i \rightarrow q: \{\text{l}_i(B_i).G_i\}$, $\alpha = p \rightarrow q_i: \text{l}_k(B_k)$. By projection rule,$proj(G) = \Gamma_1, \Gamma_2, \Gamma'$, for $\Gamma_1 = s[\{p_i\}_{i\in I}]: q \oplus \{\text{l}_i(B_i).G_i\upharpoonright_{p_i}\}_{i\in I}$, $\Gamma_2 = s[q]\!:\! \exists_{i\in I} p_i \& \{\text{l}_i(B_i).G_i\upharpoonright_q\}$, and $s[\{p_i\}_{i\in I}] \notin dom(\Gamma')$, $s[q]\notin dom(\Gamma')$. Then, by (L-select') and (L-exist), we can get $\Gamma_1 \!\xrightarrow{p_k:q\oplus \text{l}_k(B_k)}\! G_k\!\upharpoonright_{p_k}$ and $\Gamma_2 \!\allowbreak\xrightarrow{q:p_k\& \text{l}_k(B_k)} \!G_k\!\upharpoonright_q$. By (L-par) we have $\Gamma_1, \Gamma_2 \xrightarrow{p_k \rightarrow q: \text{l}_k(B_k)} G_k\upharpoonright_{p_k}, G_k\upharpoonright_q$. Since $s[\{p_i\}_{i\in I}] \notin dom(\Gamma'), s[q]\notin dom(\Gamma')$, $proj(G') = G_k\upharpoonright_{p_k}, G_k\upharpoonright_q, \Gamma'$, so we have done.

\noindent \textbf{Case G-comm'}: We can assume that $G = p\rightarrow q: \{\text{l}_i(B_i).G_i\}_{i\in I}$. Then by (G-comm'), $G' = p\rightarrow q: \{\text{l}_i(B_i).G'_i\}_{i\in I}$, $\forall i \in I. G_i \xrightarrow{\alpha} G_i'$, and $\{p, q\} \cap \textit{roles}(\alpha) = \emptyset$. So we can safely assume that $\alpha = m\rightarrow n:\text{l}_{\alpha}(B_{\alpha})$. Because $G$ is well-formed, we have $\forall r \in roles(G)$, $ G\upharpoonright_r$ is defined. So $\forall r \in roles(G)\backslash \{m, n\}$, $G\upharpoonright_r$ is defined, and $G\upharpoonright_r = G'\upharpoonright_r$, which makes up $\Gamma'$. We define that $proj(G) = \Gamma_1, \Gamma_2, \Gamma'$, where $\Gamma_1= G\upharpoonright_m$, $\Gamma_2= G\upharpoonright_n$. For $\Gamma_1$, first we know that $\forall i \in I$, $G_i \upharpoonright_m = n\oplus \text{l}_{\alpha}(B_{\alpha}).G_{i_{\alpha}}\upharpoonright_m$, and $\Gamma_1 = \sqcap \{G_i \upharpoonright_m\}_{i \in I}$. By the well-formedness and merge rule, there must be $\forall i, j \in I, \{G_i \upharpoonright_m\} = \{G_j \upharpoonright_m\}$, so $\forall i, j \in I, \{G_{i_{\alpha}} \upharpoonright_m\} = \{G_{j_{\alpha}} \upharpoonright_m\}$. Therefore, we can define $G'\upharpoonright_m = \Gamma'_1 = \sqcap \{G_{i_{\alpha}} \upharpoonright_m\}_{i \in I}$. Similarly we can get $\Gamma'_2$. Now by (L-par), we have $\Gamma_1, \Gamma_2 \xrightarrow{\alpha} \Gamma'_1, \Gamma'_2$, and $proj(G') = \Gamma'_1, \Gamma'_2, \Gamma'$, so we have done.

\noindent \textbf{Case G-exist'}:  We assume that $G = \exists_{i\in I} p_i \rightarrow q: \{\text{l}_i(B_i).G_i\}$, the proof is the same as that for case (G-comm').

\noindent \textbf{Case G-rec}: We can assume that $G = \mu t.G_0$. Then, by projection rules, $\forall r \in roles(G)\allowbreak\cup exdom(G), G\upharpoonright_r = (\mu t.G_0)\upharpoonright_r = \mu t.G_0\upharpoonright_r$, so by induction hypothesis, we have done.

Then, we prove that if $proj(G) = \Gamma = \Gamma_1, \Gamma_2, \Gamma'$, $\Gamma_1, \Gamma_2 \xrightarrow{\alpha} \Gamma_1', \Gamma_2'$, then there must exist a transition $G \xrightarrow{\alpha} G'$, and $proj(G') = \Gamma_1', \Gamma_2', \Gamma'$:

We can safely assume that $\alpha = p \rightarrow q: \text{l}_k(B_k)$. By (L-par), we have $\Gamma_1 \xrightarrow{p:q\oplus l(B)}_L \Gamma_1 '$, $ \Gamma_2 \xrightarrow{q:p\&l(B)}_L \Gamma_2 ' $. We can assume that $dom(\Gamma_1) = s[p]$, then $\Gamma_1 = s[p]:q\oplus \{\text{l}_i(B_i).T_{1_i}\}_{i\in I}$ and $\Gamma'_1 = s[p]:T_{1_k}$. Similarly, $dom(\Gamma_2) = s[q]$, $\Gamma_2 = s[q]:p\& \{\text{l}_i(B_i).T_{2_i}\}_{i\in I}$ and $\Gamma'_2 = s[q]:T_{2_k}$. 

We assume that $G = m \rightarrow n: \{\text{l}_{i_G}(B_{i_G}).G_{i_G}\}_{i_G \in I_G}$. Then, if $\{m, n\}\cap \{p, q\} \neq \emptyset$, by projection rule, $G = p \rightarrow q: \{\text{l}_{i}(B_{i}).G_{i}\}_{i \in I}$, $G\xrightarrow{\alpha}G_k$, and $proj(G') = \Gamma'_1, \Gamma'_2, \Gamma'$, we have done. Or, $\{m, n\}\cap \{p, q\} = \emptyset$. In this case, $\Gamma_1 = \sqcap \{G_{i_G} \upharpoonright_{p}\}_{i_G\in I_G}$, $\Gamma_2 = \sqcap \{G_{i_G} \upharpoonright_{q}\}_{i_G\in I_G}$. By merge rule, we have $\forall i_G, j_G\in I_G$, $G_{i_G} \upharpoonright_{p} = G_{j_G} \upharpoonright_{p}$. By projection rule, we have $\forall i_G\in I_G$, if $G_{i_G} \upharpoonright_{p} = s[p]:q\oplus \{\text{l}_{i}(B_i).G_{i_Gi}\upharpoonright_{p} \}_{i\in I}$, then $G_{i_G} \upharpoonright_{q} = s[q]:p\& \{\text{l}_{i}(B_i).G_{i_Gi}\upharpoonright_{q} \}_{i\in I}$. Therefore, $\forall i_G\in I_G$, we have $G_{i_G}\xrightarrow{\alpha}G'_{i_G}$. By (G-comm') and (G-exist'), we can get $G\xrightarrow{\alpha} G'$, where $G = m \rightarrow n: \{\text{l}_{i_G}(B_{i_G}).G'_{i_G}\}_{i_G \in I_G}$. The case when $G = \exists_{i_G \in I_G}m_{i_G} \rightarrow n: \{\text{l}_{i_G}(B_{i_G}).G_{i_G}\}$ is the same. 

Combining the two parts above, Theorem~\ref{theorem:consistency} is proved.
\qed
\end{proof}


\section{Type System}
\label{sec:type}

The type system is defined under two typing contexts: the one for channels $\Gamma$ defined in Def.~\ref{def:contextchannel}, and  $\Theta$ mapping variables to basic types and process variables to a list of types.
\[
	\Theta \triangleq \emptyset \mid \Theta, X: T_1, ..., T_n
\]
The typing judgment $\Theta\cdot \Gamma \vdash P$ states that under the typing contexts $\Theta$ and  $\Gamma$, process $P$ is well-typed. Fig.~\ref{fig:typesystemfull} presents the typing rules for MSSR $\pi$-calculus, to mainly show the difference from the standard session types~\cite{bettini2008global,scalas2019less}.

\begin{figure}[htbp]
\[
\small{
\begin{array}{ccc}
\prftree[r]{T-X}{\Theta (X) = T_1,...,T_n}{\Theta\vdash X:T_1,...,T_n}\quad 
\prftree[r]{T-Par}{\Theta\cdot \Gamma_1 \vdash P_1 \quad \Theta\cdot \Gamma_2 \vdash P_2}{\Theta\cdot \Gamma_1, \Gamma_2 \vdash P_1|P_2} \\[.2em]
\prftree[r]{T-fun}{\Theta\vdash X:S_1, S_2,...,S_n \quad \textbf{end}(\Gamma_0) \quad \Gamma_i \vdash d_i: S_i, \forall i \in \{1, 2, ..., n\}}{\Theta\cdot \Gamma_0, \Gamma_1,...,\Gamma_n \vdash X\langle d_1, d_2,...,d_n\rangle} \\[.2em]
\prftree[r]{T-def}{\Theta, X\!:\!S_1, S_2,...,S_n\cdot x_1\!:\!S_1, x_2\!:\!S_2, ...,x_n\!:\!S_n \vdash P \quad \Theta, X\!:\!S_1, S_2, ...,S_n\cdot \Gamma \vdash Q}{\Theta\cdot \Gamma \vdash \textbf{def}\ X(x_1\!:\!S_1, x_2\!:\!S_2, ..., x_n\!:\!S_n) = P\ \textbf{in}\ Q}\\[.2em]
\prftree[r]{T-branch}{\Gamma_1 \vdash c:q\&\{\text{l}_i(B_i).T_i\}_{i\in I}\quad \forall i\in I.\Theta\cdot\Gamma, y_i:B_i, c:T_i\vdash P_i \quad I \subseteq J}
{\Theta\cdot\Gamma, \Gamma_1\vdash c[q]\&\{\text{l}_i(y_i).P_i\}_{i\in J}}\\[.2em]
\prftree[r]{T-exist}{\Gamma_1 \vdash c:\exists_{i \in I}q_i\&\{\text{l}_i(B_i).T_i\}_{i\in I}\quad \forall i\in I.\Theta\cdot\Gamma, y_i:B_i, c:T_i\vdash P_i\quad I \subseteq J}
{\Theta\cdot\Gamma, \Gamma_1\vdash \exists_{i \in I}c[q_i]\&\{\text{l}_i(y_i).P_i\}_{i\in J}}\\[.2em]
\prftree[r]{T-select}{\Gamma_1 \vdash c:q\oplus \{\text{l}_i(B_i).T_i\}_{i \in I}\quad \Gamma_2\vdash d_k:B_k\quad \Theta\cdot\Gamma, c:T_k\vdash P_k \quad k \in I }
{\Theta\cdot\Gamma, \Gamma_1, \Gamma_2\vdash c[q]\oplus \{\text{l}_k\langle d_k\rangle.P_k\}} \\[.2em]
\prftree[r]{T-select'}{
	\begin{array}{ccc}
		\Gamma_1 \vdash c:\textbf{end}\quad 
		\Gamma_2 \vdash s[A]:q\oplus \{\text{l}_i(B_i).T_i\}_{i \in I}\\
  \Gamma_3\vdash d_k:B_k\quad \Theta\cdot\Gamma, c:T_k\vdash P_k \quad k \in I \quad c=s[p] \Rightarrow p \in A
	\end{array}}
{\Theta\cdot\Gamma, \Gamma_1,\Gamma_2,\Gamma_3\vdash c[q]\oplus \{\text{l}_k\langle d_k\rangle.P_k\}}\\[.2em]
\prftree[r]{T-new}{s\notin \Gamma \quad \Gamma_1 = \{s[p]: G\upharpoonright_p\}_{p \in \textit{roles}(G)} \quad \Theta \cdot \Gamma, \Gamma_1 \vdash P  }{\Theta \cdot \Gamma \vdash (\nu s) P}
\end{array}
}
\]
\caption{The Type System for Processes}
\label{fig:typesystemfull}
\end{figure}

The rules are explained as follows:
\vspace{-1mm} 
\begin{itemize}
\item Rule (T-branch): $c[q]\&\{\text{l}_i(y_i).P_i\}_{i\in J}$ is typed under $\Gamma$ and $\Gamma_1$, if under $\Gamma_1$,  $c$ has a type that is an external choice from $q$ with a smaller label set $I \subseteq J$, and for each $i\in I$, $P_i$ is typed under the channel typing context composed of $\Gamma$, the typing for bounded variable $y_i$ that occurs in $P_i$ and the continuation type $T_i$ of $c$. 

\item  Rule (T-exist): $\exists_{i \in I}c[q_i]\&\{\text{l}_i(y_i).P_i\}_{i\in J}$ is typed, if $c$ has a type that is an existential external choice from the senders in $\{q_i\}_{i \in I}$ satisfying $I \subseteq J$, and for each $i\in I$, $P_i$ is typed under $\Gamma$, the typing for $y_i$ and the continuation type $T_i$ of $c$. 

Both of the above rules indicate that processes allow more external choices. 

\item Rule (T-select): $c[q]\oplus \{\text{l}_k\langle d_k\rangle.P_k\}$ is typed, if $c$ has a type that is an internal choice towards $q$, with a label set containing $\text{l}_k$ inside.

\item Rule (T-select'):  $c[q]\oplus \{\text{l}_k\langle d_k\rangle.P_k\}$ is also typed, if the typing context corresponding to $c$ is $\textbf{end}$, but the typing context for some $s[A]$ exists such that $c$ corresponds to some role in $A$ and $c$ is fixed for performing the communications specified by $s[A]$. $A$ must be the domain of some existential branching, and the role of $c$ is among it. 

\item Rule (T-new): $(\nu s) P$ is typed under $\Gamma$, if $P$ is typed under the typing context composed of $\Gamma$ and the one for session $s$ (i.e. $\Gamma_1$). $\Gamma_1$ guarantees that there exists a global type $G$ such that each role  of $s$ has exactly the local type projected from $G$. 

The full version of the type system is given in Fig.\ref{fig:typesystemfull}. We explain the rest of them here.

\item Rule (T-X): $X$ is typed under $\Theta$ if it is in the domain of the context  $\Theta$, and its type is a list of types corresponding to the parameters of this process variable.

\item Rule (T-Par): $P_1|P_2$ is typed under a  channel typing context if it can be split into two disjoint parts such that $P_1$ and $P_2$ are typed under them resp.  

\item Rule (T-fun): $X\langle d_1, d_2,...,d_n\rangle$ is typed, if the types of the actual parameters $d_i$, $i\in \{1, 2, ..., n\}$, conform to the type of $X$, that is defined in $\Theta$.

\item Rule(T-def): $\textbf{def}\ X(x_1\!:\!S_1, x_2\!:\!S_2, ..., x_n\!:\!S_n) = P\ \textbf{in}\ Q$ is typed if $P$ is typed by assuming that  function $X$ has the type $S_1, S_2, ..., S_n$ (which are used to handle the possible occurrences of $X$ in $P$ for recursion), and furthermore, $Q$ is typed under the context for process variables extended with $X: S_1, S_2, ..., S_n$.

\end{itemize}

 \begin{example} 
We have the following typing derivation for the running example:
 
 \begin{equation*}
\small{
\dfrac{
\dfrac{
\dfrac{...}{s[{\color{myblack}R}]:G\upharpoonright_{{\color{myblack}R}}\vdash P_{{\color{myblack}r}_b}|P_{{\color{myblack}r}_d}}(\text{T-select'})\quad \dfrac{...}{s[{\color{myblack}r}_s]:G\upharpoonright_{{\color{myblack}r}_s}\vdash P_{{\color{myblack}r}_s}}(\text{T- exist})
}
{s[{\color{myblack}{\color{myblack}R}}]:G\upharpoonright_{{\color{myblack}{\color{myblack}R}}}, s[{\color{myblack}r}_s]:G\upharpoonright_{{\color{myblack}r}_s}\vdash P_{{\color{myblack}r}_s}|P_{{\color{myblack}r}_b}|P_{{\color{myblack}r}_d}}(\text{T-Par})
}{\emptyset \vdash (\nu s)(P_{{\color{myblack}r}_s}|P_{{\color{myblack}r}_b}|P_{{\color{myblack}r}_d}) }
(\text{T-new})}
\end{equation*}
The above derivation result shows that the processes are typed with the corresponding local types projected from global type $G$.

\label{example:typing}
 \end{example}

\paragraph{\textbf{Subject Reduction and Progress}}

There are two important properties for evaluating a type system: subject reduction (also called type safety) and progress.  
\begin{theorem}[Subject reduction]
	Assume $\Theta \cdot \Gamma \vdash P$ is derived according to the type system in Fig.~\ref{fig:typesystemfull}. Then, $P \to P'$ implies $\exists$ $\Gamma'$ such that $\Gamma \to_L^{\ast} \Gamma'$ and $\Theta \cdot \Gamma' \vdash P'$.
 \label{theorem:subjectreduction}
\end{theorem}

\begin{proof}
    A parallel process always starts with $(\nu s)$. By semantics, $P \to P'$ implys $(\nu s)P \to (\nu s)P'$. Assume that $\Theta \cdot \Gamma, \Gamma_1 \vdash (\nu s)P$, with $s\notin \Gamma$ and $\Gamma_1 = \{s[p]: G\upharpoonright_p\}_{p \in \textit{roles}(G)}$. Similarly, assume that $\Theta \cdot \Gamma, \Gamma_2 \vdash (\nu s)P'$,  with $s\notin \Gamma$ and $\Gamma_2 = \{s[p]: G'\upharpoonright_p\}_{p \in \textit{roles}(G')}$. By consistency, $\Gamma_1 \to_L \Gamma_2$ implies $G\to_G G'$.

    We prove the theorem by structural induction on the transition rules of processes.
    
\noindent \textbf{Case [$\&\oplus$]}: We can safely assume that $P = s[p][q]\& \{\text{l}_i(x_i).P_i\}_{i\in I} | s[q][p]\oplus \text{l}_k(\omega).Q$. Then, we know that $P\to P_k\{\omega/x_k\} | Q$. By assumption, we know $\Theta \cdot \Gamma \vdash P$, so we can assume that $\Gamma = \Gamma_1, \Gamma_2$, $\Theta \cdot \Gamma_1\vdash s[p][q]\& \{\text{l}_i(x_i).P_i\}_{i\in I}$ and $\Theta \cdot \Gamma_2\vdash s[q][p]\oplus \text{l}_k(\omega).Q$. By typing rule T-branch, $\Gamma_1 = \Gamma_3, \Gamma_5$, $ \Gamma_3 \vdash s[p]:q\& \{\text{l}_{P_i}(B_{P_i}).T_{P_i}\}_{i\in I_{\&}}$, $I_{\&} \subseteq I$. By T-select, $\Gamma_2 = \Gamma_4, \Gamma_6$, $ \Gamma_4 \vdash s[q]:p\oplus \{\text{l}_{Q_i}(B_{Q_i}).T_{Q_i}\}_{i\in I_{\oplus}}$, $\Theta\cdot \Gamma_6, \omega:B_k, s[q]:T_{Q_k} \vdash Q$, and $k \in I_{\oplus}$. Because $\exists G$ such that $s[p]:G\upharpoonright_p$ and $s[q]:G\upharpoonright_q$, by projection rules, we have $I_{\&} = I_{\oplus}$, so $k \in I_{\&}$, and $\forall i \in I_{\&}$, $\text{l}_{P_i} = \text{l}_{Q_i}$, $B_{P_i} = B_{Q_i}$. Then, we can get $\Theta\cdot \Gamma_5, x_k:B_{P_k}, s[p]:T_{P_k} \vdash P_k$, so $\Theta\cdot\Gamma_5, \Gamma_6, \omega:B_{P_k}, s[p]:T_{P_k}, s[q]:T_{Q_k}\vdash P_k\{\omega/x_k\}|Q$. By L-select and L-branch, we have $\Gamma_3 \xrightarrow{s[p]:q\&\text{l}_k(B_{x_k})}s[p]:T_{P_k} $ and $\Gamma_4 \xrightarrow{s[q]:p\oplus \text{l}_k(B_{x_k})} s[q]:T_{Q_k}$. Then by L-par, we can get $\Gamma_3, \Gamma_4 \xrightarrow{s[p][q]\text{l}_k(B_{x_k})}s[p]:T_{P_k}, s[q]:T_{Q_k}$. So we let $\Gamma' = \Gamma_5, \Gamma_6, \omega:B_{P_k}, s[p]:T_{P_k}, s[q]:T_{Q_k}$, then we have $\Gamma \rightarrow_L \Gamma'$ and $\Theta \cdot \Gamma' \vdash P_k\{\omega/x_k\} | Q$. By theorem~\ref{theorem:consistency}, exists $G \xrightarrow{\alpha}G'$, such that $G'$ and $\Gamma'$ is consistent.

\noindent \textbf{Case [$\exists\oplus$]}: Similarly, we can safely assume that $P = s[p]\exists\{[q_i]\text{l}_i(x_i).P_i\}_{i\in I} | s[q_k][p]\oplus \text{l}_k(\omega).Q$. Then, we know that $P\to P_k\{\omega/x_k\} | Q$. By assumption, we know $\Theta \cdot \Gamma \vdash P$, so we can assume that $\Gamma = \Gamma_1, \Gamma_2$, $\Theta \cdot \Gamma_1\vdash s[p]\exists\{[q_i]\text{l}_i(x_i).P_i\}_{i\in I}$ and $\Theta \cdot \Gamma_2\vdash s[q_k][p]\oplus \text{l}_k(\omega).Q$. By typing rule T-exist, $\Gamma_1 = \Gamma_3, \Gamma_5$, $ \Gamma_3 \vdash s[p]:\exists_{i\in I_{\&}} q_i\&\{{\text{l}_{P_i}} {(B_{P_i})}\allowbreak .{T_{P_i}}\}$, $I_{\&}\! \subseteq\! I$. By T-select', let $A = \{q_i\}_{i\in I_{\oplus}}$, $\Gamma_2 = \Gamma_4, \Gamma_6, \Gamma_7$, $ \Gamma_4 \vdash s[A]:p\oplus \{\text{l}_{Q_i}(B_{Q_i}).T_{Q_i}\}_{i\in I_{\oplus}}$, $\Theta\cdot \Gamma_6, \omega:B_k, s[A]:T_{Q_k} \vdash Q$,$\Gamma_7\vdash s[p_k]:\textbf{end}$, and $k \in I_{\oplus}$. Because $\exists G$ such that $s[p]:G\upharpoonright_p$ and $s[A]:G\upharpoonright_A$, by projection rules, we have $I_{\&} = I_{\oplus}$, so $k \in I_{\&}$, and $\forall i \in I_{\&}$, $\text{l}_{P_i} = \text{l}_{Q_i}$, $B_{P_i} = B_{Q_i}$. Then, we can get $\Theta\cdot \Gamma_5, x_k:B_{P_k}, s[p]:T_{P_k} \vdash P_k$, so $\Theta\cdot\Gamma_5, \Gamma_6, \Gamma_7, \omega:B_{P_k}, s[p]:T_{P_k}, s[A]:T_{Q_k}\vdash P_k\{\omega/x_k\}|Q$. By L-select' and L-exist, we have $\Gamma_3 \xrightarrow{s[p]:q_k\&\text{l}_k(B_{x_k})}s[p]:T_{P_k} $ and $\Gamma_4 \xrightarrow{s[A]:p\oplus \text{l}_k(B_{x_k})} s[q_k]:T_{Q_k}$. Then by L-par, we can get $\Gamma_3, \Gamma_4 \xrightarrow{s[p][q_k]\text{l}_k(B_{x_k})}s[p]:T_{P_k}, s[q_k]:T_{Q_k}$. So we let $\Gamma' = \Gamma_5, \Gamma_6, \Gamma_7, \omega:B_{P_k}, s[p]:T_{P_k}, s[q_k]:T_{Q_k}$, then we have $\Gamma \rightarrow_L \Gamma'$ and $\Theta \cdot \Gamma' \vdash P_k\{\omega/x_k\} | Q$. By theorem~\ref{theorem:consistency}, exists $G \xrightarrow{\alpha}G'$, such that $G'$ and $\Gamma'$ is consistent.

\noindent \textbf{Case [$Ctx$]}: We can safely assume that $P = \mathbb{C}[Q], P' = \mathbb{C}[Q'] , Q\to Q'$, then we know that $P \to P'$. By assumption, we know $\Theta \cdot \Gamma \vdash P$, so we can assume that $\Gamma = \Gamma_1, \Gamma_2$, $\Theta \cdot \Gamma_1 \vdash Q$ and $\Theta \cdot \Gamma_2 \vdash \mathbb{C}[]$.
By the induction hypothesis, $\exists \Gamma_1 \to \Gamma'_1$, s.t. $\Theta \cdot \Gamma'_1 \vdash Q'$. So we can get $\Theta \cdot \Gamma'_1, \Gamma_2 \vdash \mathbb{C}[Q']$.

\qed
\end{proof}

 The subject reduction guarantees that if a typed process takes a step of evaluation, then the resulting process is also typed. However, it 
 does not guarantee that a well-typed process can always take one further step of execution if it is not terminated: it could be stuck in a deadlock while waiting for dual actions. The
deadlock freedom is usually called
\emph{progress} in some literature~\cite{DBLP:books/daglib/0005958}.  
We prove the following theorem on progress. 
\begin{theorem}[Progress]
 Let $P$ be $(\nu s)(\Pi_{i \in I} P_i)$, with $|\textit{roles}(P_i)| = 1$ and $\textit{roles}(P_i) \cap \textit{roles}(P_j) = \emptyset$ for any $i, j \in I, i \neq j$. Then  $\Theta, \emptyset \vdash (\nu s) P$ and  $P\to ^* P' \nrightarrow$ implies $P' = \textbf{0}$. Here $\textit{roles}(P_i)$ returns the roles of $P_i$. 
 \label{theorem:progress1}
	\end{theorem} 

 \begin{proof}
    By typing rule (T-new), $\Theta \cdot \Gamma ,\Gamma_1\vdash(\Pi_{i \in I} P_i)$ ,where $\Gamma_1 = \{s[p]: G\upharpoonright_p\}_{p \in \textit{roles}(G)}$ and $s\notin \Gamma$. Assume that $P \neq 0$, then $G\neq \textbf{end}$. Therefore, $G = \alpha.G'$, where $\alpha$ is a global interaction. Because $|\textit{roles}(P_i)| = 1$ and $\textit{roles}(P_i) \cap \textit{roles}(P_j) = \emptyset$ for any $i, j \in I, i \neq j$, $\alpha$ is executable, which means that $P \xrightarrow{\alpha}$. So in this case, $P \nrightarrow$ implies $P = 0$.
     \qed
 \end{proof}

 However, the above theorem has a strong restriction by requiring that each sequential process plays only one role. If a   process plays multiple roles in a session, deadlock might occur. The reason is that the projection from a global type to local types loses the orders of actions occurring over different roles. See the example below.

 \begin{example}
     Define $P = (\nu s) P_1 | P_2$, where 
    $P_1 = s[p][r]\&\{\text{l}_1(x_1).s[p][q]\oplus\{\text{l}_2(a_2).\zero\}\}$
  and $P_2 = s[q][p]\&\{\text{l}_2(x_2).s[r][p]\oplus\{\text{l}_1(a_1).\zero\}\}$, where     
 $P_1$ plays one role $p$, and $P_2$ plays two roles $q$ and $r$. Let 
 $G = r \!\rightarrow\! p\!:\! \{\text{l}_1(B_1). p\!\rightarrow\! q\!:\! \{\text{l}_2(B_2).\textbf{end}\}\}$, $\Gamma \vdash a_1 \!:\! B_1, a_2\! :\! B_2$, we have 
 \begin{equation*}
     \begin{array}{c}
G\upharpoonright_r = p \oplus\text{l}_1(B_1).\textbf{end}
 \quad
G\upharpoonright_q = p \&\text{l}_2(B_2).\textbf{end}
 \quad
 G \upharpoonright_p = r\&\text{l}_1(B_1).  q\oplus \text{l}_2(B_2).\textbf{end}
     \end{array}
 \end{equation*} 
 By the type system,
 $\emptyset \cdot \Gamma \vdash P$. However, $P$ leads to a deadlock.
 In next section, we will define a communication type system by abandoning the restriction on a single role. 
 \label{example:deadlock3}
 \end{example}

\paragraph{\textbf{Session Fidelity}}
The session fidelity connects process reductions to
typing context reductions. It says that the interactions of a typed process follow the transitions of corresponding local types, thus eventually follow the protocol defined by the global type. 

\begin{theorem}[Session Fidelity]
Let $P$ be $\Pi_{i \in I} P_i$, with $|\textit{roles}(P_i)| = 1$ and $\textit{roles}(P_i) \cap \textit{roles}(P_j) = \emptyset$ for any $i, j \in I, i \neq j$. If $\Theta, \Gamma \vdash P$, $\Gamma = \{s[p]: G\upharpoonright_p\}_{p \in \textit{roles}(G)}$ for some global type $G$, and $\Gamma \to_L$,  then there exist $\Gamma'$ and $P'$ such that $\Gamma \to_L \Gamma'$, $P\to  P'$ and $\Theta \cdot \Gamma' \vdash P'$.
\label{theorem:session_fidelity}
\end{theorem}

\begin{proof}
    By rule (L-par), $\Gamma \to_L$ implies $\exists \Gamma_1, \Gamma_2$, s.t. $\Gamma = \Gamma_1, \Gamma_2, \Gamma'$, $\Gamma_1 \xrightarrow{p:q\oplus l(B_1)}_L$ and $\Gamma_2 \xrightarrow{q:p\&l(B_2)}_L$. We can safely assume that $dom(\Gamma_1) = s[p]$(for $s[A]$, the proof is the same), then $\Gamma_1$ is derived by (T-select), and $\Gamma_2$ is derived by (T-branch). We let $\Gamma_1 =  s[p]:q\oplus \{\text{l}_{i_p}(B_{i_p}).T_{i_p}\}_{i_p \in I_p}$,  $\Gamma_2 = s[q]:p\&\{\text{l}_{i_q}(B_{i_q}).T_{i_q}\}_{i_q\in I_q}$, $k\in I_p\subseteq I_q$, $l = \text{l}_{k_p} = \text{l}_{k_q}$, $B = B_{k_p} = B_{k_q}$.

    Then, because $|\textit{roles}(P_i)| = 1$ and $\textit{roles}(P_i) \cap \textit{roles}(P_j) = \emptyset$, we have $P = P_p|P_q|P'$, where $P_p$ only plays role $p$, and $P_q$ only plays role $q$. That means $P_p = c[q]\oplus \{\text{l}_{k'}\langle d_{k'}\rangle.P_{p_{k'}}\}$, $P_q = c[q]\&\{\text{l}_i(y_i).P_{q_i}\}_{i\in J_q}$ by typing rule (T-select) and (T-branch), where $k' \in I_p$ and $I_q \subseteq J_q$. For $I_p\subseteq I_q$, we can get $k' \in J_q$. Therefore, by semantics [$\&\oplus$], we have $P_p|P_q \to P_{p_{k'}}|P_{q_{k'}}$, and $\Theta \cdot s[p]:T_{{k'}_p}, s[q]:T_{{k'}_q}\vdash P_{p_{k'}}|P_{q_{k'}}$.
\qed
\end{proof}

\section{A Communication Type System For Progress}
\label{sec:communication}

This section defines a communication type system for guaranteeing the progress of processes in MSSR $\pi$-calculus, especially for the cases when one sequential process plays multiple roles. This is achieved by recording the execution orders of communication events for processes and then checking if there exists a dependence loop indicating that all processes are on waiting status and cannot proceed.

\paragraph{\textbf{Notations}}
First of all, we introduce the notion of \emph{directed events}:
\vspace{-1.5mm}
\begin{itemize}
    \item A \emph{directed event} $(c[r], \text{l}, i)$, where $c$ is a channel, $r$ a role, $\text{l}$ a label and $i>0$ a time index, representing that this event is the $i$-th communication action of $c[r]$ and the label is $\text{l}$. The event is \emph{directed} from the role of $c$ to $r$, e.g. if $c=s[p]$, it means that $p$ sends to or receives from $r$ at this event.

    \item A \emph{directed event with channel mobility} $(c[r], \text{l}, i, c', A)$, where $c'$ is the channel exchanged over this event    
    and $A$ is a set recording the communication history of $c'$ till this event occurs.  Each element of $A$ has form $(c'[p], k)$, where $p$ is a role and $k > 0$, indicating that $k$ times of communications on $c'[p]$ have occurred before $c'$ is moved by this event.      
\end{itemize}
Below we denote the sets of these two kinds of events by $\mathcal{E}$ and $\mathcal{E}_M$, and for simplicity, call them \emph{events} and \emph{mobility events} resp. For any  $e$ in $\mathcal{E}$ or $\mathcal{E}_M$, 
we will use $e.i$ to represent the $i$-th sub-component.

Each communication typing rule takes the form $\Delta \vdash P \triangleright U, M, R$, where  $P$ is the process to be typed, $U \subseteq \mathcal{E}$ is the set of the least events of $P$ such that no event has occurred prior to them, $M$ is the set of mobility events that have occurred in $P$, $R \subseteq \mathcal{E} \times \mathcal{E}$ is a set of relations describing the execution order of events, and $\Delta$ is the communication context for functions, of the form $\{X(\tilde{y}) \triangleright U_X, M_X, R_X\}_{X \in P}$, for any function $X$ called in $P$.
$(e_1, e_2) \in R$ means that $e_1$ occurs strictly before $e_2$ in $P$, and we will write  $e_1 \prec_R e_2$ as an abbreviation ($R$ omitted without confusion). The addition of $M$ is mainly to transfer the history information of each channel to be moved to the corresponding target receiver.

Before presenting the typing rules, we introduce some notations (all symbols are defined as before).   $U\removal c[r]$ removes all events of $c[r]$ from $U$, i.e. $\{e \mid e\in U \wedge e.1 \neq c[r]\}$; 
$pre(x, U)$ is defined as $\{x \prec y \mid y \in U\}$, to represent that $x$ occurs prior to all events in $U$; $\mathcal{C}(d)$ is a boolean condition judging whether  $d$ is a channel value; $M_1 \lhd b \rhd M_2$ returns $M_2$ if $b$ is true, otherwise $M_1$; events
 $(s[p][q], \text{l}, i)$ and $(s[q][p], \text{l}, i)$ are said to be \emph{dual}, and we use $\overline{e}$ to represent the dual event of $e$. At the end, we define the \emph{$k$-th index increase} of events with respect to a fixed pair of sender and receiver $\alpha$, which has the form $c[q]$ for some $c$ and $q$.  Let $e$ be an event $(c[r], l, i)$ and $me$ be a mobility event $ (c[r], l, i, c', A)$:
\[
\begin{array}{lll}
e\uparrow_k^\alpha \triangleq\!
\left \{
\begin{array}{cc}
 (c[r], l, i+k)   &  \mbox{if $\alpha = c[r]$} \\
 e    & \mbox{otherwise}
\end{array}\right. 
\quad 
A\uparrow^\alpha_k \triangleq\!
\left \{
\begin{array}{cc}
 A \!\setminus\! (\alpha, i) \cup \{(\alpha, i+k)\}   &  \mbox{if $(\alpha, i) \in A$} \\
 A \cup \{(\alpha, k)\}    & \mbox{otherwise}
\end{array}\right.
\\
me\uparrow_k^\alpha \triangleq
\left \{
\begin{array}{lll}
 (c[r], l, i+k, c', A)   &  \mbox{if $\alpha = c[r]$} \\
 (c[r], l, i, c', A\uparrow_k^{c'[p]})   &  \mbox{if $\alpha = c'[p]$ for some $p$} \\
 me    & \mbox{otherwise}
\end{array}\right. \\
R\uparrow_k^{c[q]} \triangleq \{e\uparrow_k^{c[q]}\prec e'\uparrow_k^{c[q]} | e \prec e' \in R\} \quad
M \uparrow_k^{c[q]} \triangleq \{me\uparrow_k^{c[q]} | me  \in M\}\\
R \uparrow^A \triangleq \bigcup_{(c[q], k) \in A}R\uparrow_k^{c[q]}
\qquad\qquad\qquad M \uparrow^A \triangleq \bigcup_{(c[q], k) \in A}M\uparrow_k^{c[q]}
\end{array}
\]
$e\uparrow_k^\alpha$ says, when $e$ is occurring on $\alpha$, its time index is increased by $k$, otherwise not changed. 
$me\uparrow^\alpha_k$ not only promotes the event itself with respect to $\alpha$ (the first case), but also the communication history of $c'$   (the second case). As defined by $A\uparrow_k^{\alpha}$, the pair corresponding to $\alpha$ is increased by $k$, or added if it is not in $A$.   $R\uparrow_k^{c[q]}$ and $M \uparrow_k^{c[q]}$ perform the pointwise $k$-th index increase with respect to $c[q]$ to all events in them, $R \uparrow^A$ and $M \uparrow^A$ perform the  index increase of events in $R$ and $M$ with respect to $A$. For simplicity, we omit the subscript $k$ of all the above definitions when $k = 1$.

\begin{figure}[t]
\[
\small{
\begin{array}{ccc}

\prftree[r]{C-select}{\Delta \vdash P \triangleright U, M, R \quad M'= (\emptyset \lhd \mathcal{C}(d) \rhd  \{(c[q], \text{l}, 1, d, \emptyset)\})}
{\begin{array}{ll}
     \Delta\vdash c[q]\oplus \{\text{l} \langle d \rangle.P \}
\triangleright &  \{(c[q], \text{l}, 1)\}, M\uparrow^{c[q]} \cup M', \\
&\  R\uparrow^{c[q]}  \cup \{pre((c[q], \text{l}, 1), U \removal c[q])\}
\end{array}} 
\\[.4em]
\prftree[r]{C-bran}{\forall i\in I \quad \Delta \vdash P_i \triangleright U_i, M_i, R_i 
\quad  M_i'= (\emptyset \lhd \mathcal{C}(y_i) \rhd  \{(c[q], \text{l}_i, 1, y_i, \emptyset)\})}
{\begin{array}{ll}
     \Delta\vdash c[q]\&\{\text{l}_i(y_i).P_i\}_{i\in I} \triangleright &\bigcup_{i\in I}\{(c[q], \text{l}_i, 1)\}, \bigcup_{i\in I}({M_i}\uparrow^{c[q]}\cup M_i'),\\
     & \bigcup_{i\in I}({R_i\uparrow^{c[q]}} \cup \{pre((c[q], \text{l}_i, 1), U_i\removal c[q])\})
\end{array}} \\[.5em]
\prftree[r]{C-exist}{\forall i\in I \quad \Delta \vdash P_i \triangleright U_i, M_i, R_i 
\quad  M_i'= (\emptyset \lhd \mathcal{C}(y_i) \rhd  \{(c[q_i], \text{l}_i, 1, y_i, \emptyset)\})}
{\begin{array}{ll}
     \Delta\vdash \exists_{i \in I}c[q_i]\&\{\text{l}_i(y_i).P_i\} \triangleright &\bigcup_{i\in I} \{(c[q_i], \text{l}_i, 1)\}, \bigcup_{i\in I}(M_i\uparrow^{c[q_i]}\cup M_i'),\\
     & \bigcup_{i\in I}({R_i\uparrow^{c[q_i]}} \cup \{pre((c[q_i], \text{l}_i, 1), U_i\removal c[q_i])\})
\end{array} }\\[.8em]
\prftree[r]{C-0}{}{\Delta \vdash \textbf{0} \triangleright \emptyset, \emptyset, \emptyset} 
\quad \prftree[r]{C-par}{\Delta \vdash P  \triangleright U, M, R \quad \Delta \vdash Q  \triangleright U', M', R'}
{\Delta\vdash  P | Q \triangleright U \cup U', M\cup M', R \cup R' }
\\[.4em]
\prftree[r]{C-call}{ }{\Delta, X(\tilde{x}) \triangleright U, M, R \vdash X(\tilde{d}) \triangleright U[\tilde{d}/\tilde{x}], M[\tilde{d}/\tilde{x}], R[\tilde{d}/\tilde{x}]}
\\[.6em]
\prftree[r]{C-def}{\Delta \vdash P[\textbf{0}/X] \triangleright  U, M, R \quad  \Delta, X(\tilde{x}) \triangleright U, M, R \vdash Q \triangleright  U', M', R'}{\Delta \vdash \textbf{def}\ X(\tilde{x}) = P\ \textbf{in}\ Q \triangleright U', M', R'}\\[.2em]
\prftree[r]{C-unify}{\Delta \vdash P  \triangleright U, M, R \quad  M, R \Rightarrow_u M', R'}
{\Delta\vdash  P  \triangleright U, M', R'} \quad 
\prftree[r]{C-trans}{\Delta \vdash P  \triangleright U, M, R \quad  R \Rightarrow_t R'}
{\Delta\vdash  P  \triangleright U, M, R'}
\\ 
\hline 
\\
\prftree[r]{Unify}{B = A[x/s[m]]}{M \cup  (s[p][q], \text{l}, i, s[m], A) \cup (s[q][p], \text{l}, i, x, \emptyset), R
       \Rightarrow_u
        (M\uparrow^B)[s[m]/x], (R\uparrow^B)[s[m]/x]}\\[.3em]
\prftree[r]{Trans-1}{\{e_1 \prec e_2, \overline{e_2} \prec e_3\} \Rightarrow_t  \{e_1 \prec e_3\}} \quad
\prftree[r]{Trans-2}{R' \subseteq R \quad R' \Rightarrow_t R''}{R \Rightarrow_t R \cup R''}
\end{array}
}
\]
\caption{A Communication Type System for Processes and Auxiliary Definitions}
\label{fig:typesystem2}
\end{figure}


\subsection{Typing Rules}
Fig.~\ref{fig:typesystem2} presents the communication type system for MSSR session calculus and the auxiliary rules for unification and transitivity  of relations, with explanations below: 
\begin{itemize}
    \item Rule (C-select): The prefix selection produces event $(c[q], \text{l}, 1)$, which becomes the (only) least event of the current process. If $d$ is a channel value,  $(c[q], \text{l}, 1, d, \emptyset)$ is added to $M$, especially $\emptyset$ indicating no communication along $d$ has occurred till now. The occurrence of $(c[q], \text{l}, 1)$ promotes the time indexes of events of $c[q]$  in both $R$ and $M$   by 1. The precedence from $(c[q], \text{l}, 1)$ to all events of non-$c[q]$ channels in $U$ holds and the corresponding relations are added to the $R$-set.  

    \item Rule (C-bran): The prefix branch process produces a set of least events $(c[q], \text{l}_i, 1)$ for  $i\in I$. When $y_i$ is a channel, $(c[q], \text{l}_i, 1, y_i, \emptyset)$ is added to $M_i$. The $M_i, R_i$-sets are joined together, and furthermore, they promote all events with respect to $c[q]$. The precedence  from $(c[q], \text{l}_i, 1)$  to non-$c[q]$ events in $U_i$ are added to each $R_i$-set. 

    \item Rule (C-exist): Each existential branch from $q_i$ for $i\in I$ is handled similarly as in (C-bran) to update the three sets.

    \item Rule (C-par): The three communication typing sets corresponding to $P$ and $Q$ are joined together for $P | Q$. 

    \item Rules (C-call) and (C-def): Function call $X(\tilde{d})$ instantiates the communication typing sets for $X(\tilde{x})$ by substituting $\tilde{d}$ for $\tilde{x}$. Function declaration $X(\tilde{x})$ unfolds the function body $P$ once by replacing the (possibly recursive) occurrence of $X$ in $P$ by $\textbf{0}$ to compute the communication typing sets of $X(\tilde{x})$. This is adequate, because by the type system of Fig.~\ref{fig:typesystemfull}, each process is typed with the corresponding projection of a global type, thus all end-point processes conform to the same recursive structure as specified by the global type and one-time unfolding is enough for representing all cases of communication orders between events. 

    \item Rule (C-unify): It defines the unification of $M$ and $R$ according to rule (Unify). Rule (Unify) transfers the channel value $s[m]$ and its communication history record\-ed by $A$ from sender $p$ to receiver $q$, by first promoting the events on $x$ with respect to the history $A[x/s[m]]$ and then substituting all occurrences of $x$ for $s[m]$. Through unification, all variable channels in $R$ are instantiated to constant channels and the events of them have the correct time indexes (globally). 

    \item Rule (C-trans): It defines the transitivity of $R$ according to rules  (Trans-1) and (Trans-2), based on the synchronization of an event and its dual.

\end{itemize}





A relation $R$ is \emph{safe}, if $e \prec \overline{e} \notin R$ for any $e$, i.e. there does not exist an event $e$ such that $e$ and its dual event $\overline{e}$ (that are expected to be synchronized) are in a strict order. 
\begin{theorem}[Progress]
    Suppose $\Theta, \emptyset \vdash (\nu s) P$ and
    $\Delta \vdash P \triangleright U, M, R$, if 
    both (C-unify) and (C-trans) cannot be applied anymore to $M$ and $R$,  and if $R$ is safe, then  $P\to ^* P' \nrightarrow$ implies $P' = \textbf{0}$.
    \label{theorem:progress2}
\end{theorem}

\begin{proof}
First given a process $Q$, we write $Q \xrightarrow{s[p][q]\odot l(d))}$ with $\odot \in \{\oplus, \&\}$ if $Q$ is prefixed with the corresponding event, e.g. $c[q]\& \{\text{l}_i(x_i).P_i\}_{i\in I} \xrightarrow{c[q]\&\text{l}_i(x_i)}$ for each $i \in I$. Suppose $P' \nrightarrow$ and $P' \neq \textbf{0}$, then there must exist $P_1, \cdots, P_k$  such that $P' = P_1 | P_2 | \cdots | P_k$ and each $P_i$ must be prefixed processes. From $\Delta \vdash P \triangleright U, M, R$ and (C-Par), we obtain $\forall i \in I$, $\Delta \vdash P_i \triangleright U_i, M_i, R_i$. Given a relation $R$, we define $\mathcal{E}_R = \{e|e\prec e' \in R\ or\ e'\prec e \in R\}$ to be the events that occur in $R$, then $R$ is a partial order relation on $\mathcal{E}_R$. Suppose $P_1 \xrightarrow{s[p_1][q_1]\odot \text{l}_1(d_1)}$, there must exist $k\neq 1$ such that $P_k \xrightarrow{s[p'_k][q'_k]\odot \text{l}_k(d_k)}\longrightarrow^* \xrightarrow{s[q_1][p_1]\overline{\odot} \text{l}_1(d_1')}$, and we have $\{p'_k, q'_k\} \allowbreak\neq \{p_1, q_1\}$ for the session type system. Now we expect to get $((s[p'_k][q'_k], \text{l}_k, i_k) \prec \allowbreak(s[q_1][p_1], \text{l}_1, i_1)) \in R_k$ for some $i_k$ and $i_1$ (corresponding to their time indexes in the original $P$).

Here we examine each case of {$P_k$}.

\noindent \textbf{Case} $P_k = s[p'_k][q'_k]\oplus\{\text{l}_k \langle d \rangle.P'_k \}$: By (C-select), $\Delta \vdash P'_k \triangleright U'_k, M'_k, R'_k$, and $\{pre(\allowbreak(s[p'_k][q'_k], \text{l}_k, i_k), U'_1 \removal s[p'_k][q'_k])\} \subseteq R_k$. So we know that $(s[p'_k][q'_k], \text{l}_k, i_k)\in \mathcal{E}_{R_k}$. Similarly, for $P_k \xrightarrow{s[p'_k][q'_k]\odot \text{l}_k(d_k)}P'_k\longrightarrow^* \xrightarrow{s[q_1][p_1]\overline{\odot} \text{l}_1(d_1')}$, we can get $(s[q_1][p_1], \text{l}_1, i_1)\allowbreak\in \mathcal{E}_{R'_k}$. Observing the form $Q\xrightarrow{e_1} Q'\xrightarrow{e_2} Q''$, we can see that, according to the rules of the type of communication, $e_2$ is in the $U'_Q$ corresponding to $Q'$, and thus $e_1\prec e_2$ is in the $R_Q$ corresponding to $Q$. Recursively applying this conclusion in the above equation, we can know that $(s[p'_k][q'_k], \text{l}_k, i_k) \prec (s[q_1][p_1], \text{l}_1, i_1)\in R_k$.

\noindent \textbf{Case} $P_k = s[p'_k][q'_k]\&\{\text{l}_{ki}(y_{ki}).P'_{ki}\}_{i\in J}$: We assume that $l_{kj} = l_k$. Here by the semantics, we can get $P_k \xrightarrow{s[p'_k][q'_k]\odot \text{l}_k(y_{kj})}P'_k[d_k/y_{kj}]\longrightarrow^* \xrightarrow{s[q_1][p_1]\overline{\odot} \text{l}_1(d_1')}$. By (C-bran), $\Delta \vdash P'_k \triangleright U'_k, M'_k, R'_k$, and $\{pre(\allowbreak(s[p'_k][q'_k], \text{l}, i_k), U'_1 \removal s[p'_k][q'_k])\} \subseteq R_k$. If $s[q_1][p_1]$ is not transmitted through a mobility event, then similar to the previous case, we have $(s[p'_k][q'_k], \text{l}_k, i_k) \prec (s[q_1][p_1], \text{l}_1, i_1)\in R_k$. If $s[q_1][p_1]$ is transmitted through a mobility event, $\mathcal{C}(y_{kj})$ is true. By (C-unify) and (Unify), we can get $A, B$, s.t. $B = A[x/s[q_1][p_1]]$, and there is a mobility event $me \in M$, s.t. $me = (\alpha_k, \text{l}, i_k,\allowbreak s[q_1][p_1], A)$. Notice that $\forall \alpha, i$, $(\alpha, i) \in A \implies \exists j$, $(\alpha, j) \in A\uparrow^\alpha_k$. So for $me\in M$, we can find $j$, s.t. $(s[q_1][p_1], j)\in A$, then $(x, j)\in B$. If $\alpha_k = s[p'_k][q'_k]$, then we can get $(x, \text{l}_1, i_1)\in \mathcal{E}_{R_k'\uparrow^B}$. By replacing $x$ by $s[p_1][q_1]$, we have $((s[p'_k][q'_k], \text{l}_k, i_k) \prec (s[q_1][p_1], \text{l}_1, i_1))\in R_k$. Otherwise, $(s[p'_k][q'_k], l_1, i_1)\in \mathcal{E}_{R_k'\uparrow^B}$, we can use the rule (C-unify) recursively until the channel $s[p_1][q_1]$ is transmitted, so the conclusion still holds.

\noindent \textbf{Case} $P_k = \exists_{i \in I}s[p'_k][q'_{ki}]\&\{\text{l}_{ki}(y_{ki}).P'_{ki}\}$: Replacing $q'_{k}$ by $q'_{ki}$, the proof is the same as that for case $P_k = s[p'_k][q'_k]\&\{\text{l}_{ki}(y_{ki}).P'_{ki}\}_{i\in J}$.

Now we obtain $(s[p'_k][q'_k], \text{l}_k, i_k) \prec (s[q_1][p_1], \text{l}_1, i_1) \in R_k$ for $i_k$ and $i_1$. Now consider the prefix event of $P_k$, there must exist $j \neq k$ such that $P_j \xrightarrow{s[p''_j][q''_j]\odot \text{l}_j(d_j)}\longrightarrow^*\allowbreak \xrightarrow{s[q'_k][p'_k]\overline{\odot} \text{l}_k(d_k')}$. If $j=1$, $p''_j = p_1$ and $q''_j = q_1$, we have $(s[p_1][q_1], \text{l}_1, i_1) \prec (s[q_k][p_k],\allowbreak \text{l}_k, i_k)$, then a reflexive relation is obtained for both $s[p_1][q_1]$ and $s[q_1][p_1]$, $s[p_k][q_k]$ and $s[q_k][p_k]$. If $j \neq 1$,  we look for the dual event of prefix of $P_j$, and repeat the process recursively till we find the dual event of some prefix in the already traversed process. A precedence path that reaches a reflexive relation between two events dual to each other by applying (Trans-1) and (Trans-2) is formed. Denote the path by $\sigma$.

For the events along the path $\sigma$, it belongs to original $P$  (e.g. $s[p][q]\odot l(c)$ with $c$ as a constant), or is instantiated from an event of $P$ containing variables (e.g. $x\odot l(y)$ with $x, y$ as variables). The event of the first case is contained in $R$, while the event of the second case is also contained in $R$ after it applies the rule (Unify) to instantiate the variable subjects of events as constant session channels. If these events are not from the recursive call of some $X$, the path belongs to $R$ and the contradiction is obtained to the fact that $R$ is safe. 

If some event is from the recursive call of some $X$, whole events are not included in $R$ according to rule (C-Def). We consider a recursive call to $X_1$ in a thread $p_1$. For $p_1$, suppose that the set of all events prior to the call to $X_1$ is $A_1$, and the set of all events at the i-th recursive execution of $X_1$ is $B^i_1$. Then, $\forall i, j, a , b_i, b_j$, $i < j$, $a \in A_1$, $b_i \in B^i_1$, $b_j \in B^j_1$, we have $a\prec b_i \prec b_j$. For events in $X_1$, from the fact that process $p_1$ is well-typed with types projected from a global type, we know that their dual events must also be located in the recursive call of the thread in which they are located. Similarly, the above conclusion holds for these threads. Thus, if there exists a loop such that event $e$ in $X_1$ satisfies $e\prec \bar{e}$, then all events contained in this loop must be in recursive calls. This means that if P is safe when it does not contain the part of this recursive call ($\Delta \vdash P[\textbf{0}/X] \triangleright  U, M, R$), and the part of this recursive call is also safe ($\Delta, X(\tilde{x}) \triangleright U, M, R \vdash Q \triangleright U', M', R'$), then the proof still holds.
$\qed$
\end{proof}

\begin{example}
We show how to use the communication type system to type the processes in Exam.~\ref{example:deadlock3}.
In particular, $P_2$ produces a relation between $s[q][p]$ and $s[r][p]$, which is absent in the previous type system. 
\small{
\begin{equation*}
    \dfrac{
        \dfrac{
        \Delta \vdash \zero \triangleright \emptyset, \emptyset, \emptyset
        }{
        \Delta \vdash s[p][q]\oplus\{\text{l}_2(a_2).\zero\} \triangleright \{(s[p][q], \text{l}_2, 1)\}, \emptyset, \emptyset
        }[\text{C-select}]
    }{
    \begin{array}{ll}
        \Delta \vdash s[p][r]\&\{\text{l}_1(x_1).s[p][q]\oplus\{\text{l}_2(a_2).\zero\}\} \triangleright & \{(s[p][r], \text{l}_1, 1)\}, \emptyset,  \\
         & \{(s[p][r], \text{l}_1, 1) \prec (s[p][q], \text{l}_2, 1)\}
    \end{array}
    }[\text{C-bran}]
\end{equation*}
}
\small{
\begin{equation*}
    \dfrac{
        \dfrac{
        \Delta \vdash \zero \triangleright \emptyset, \emptyset, \emptyset
        }{
        \Delta \vdash s[r][p]\oplus\{\text{l}_1(a_1).\zero\} \triangleright \{(s[r][p], \text{l}_1, 1)\}, \emptyset, \emptyset
        }[\text{C-select}]
    }{
    \begin{array}{ll}
         \Delta \vdash s[q][p]\&\{\text{l}_2(x_2).s[r][p]\oplus\{\text{l}_1(a_1).\zero\}\} \triangleright & \{(s[q][p], \text{l}_2, 1)\}, \emptyset,  \\
         & \{(s[q][p], \text{l}_2, 1) \prec (s[r][p], \text{l}_1, 1)\}
    \end{array}
    }[\text{C-bran}]
\end{equation*}
}


By [C-par],  (Trans-1) and (C-Trans),  $P_1|P_2$ produces a relation set containing
$ (s[p][r], \text{l}_1, 1) \allowbreak\prec (s[r][p], \text{l}_1, 1)$, which is not \emph{safe}.


\label{example:comm-deadlock}
\end{example}

\section{Modelling Rust Multi-threads in MSSR Session Types}
\label{sec:application}


Rust offers both message-passing and shared-state mechanisms for concurrent programming. This section presents how to model these different concurrency primitives using MSSR session types.

\paragraph{\textbf{Channel}}
Rust library provides communication primitives for handling multi-producer, single-consumer communications. A channel $x$ contains two parts 
$(tx, rx)$, with $tx$ standing for the transmitter and $rx$ for the receiver. The channel can be synchronous or asynchronous, and allows a buffer to store the values. Let $s$ be the session  created for $(tx, rx)$, $p$ and $q$ the owner threads of $tx$ and $rx$ respectively. 
If the channel is synchronous with buffer size 0, it is consistent with synchronous communication in session types. The Rust primitives $tx.\texttt{send}(val)$ and $a = rx.\texttt{recv}()$, representing that $tx$ sends $val$ and $rx$ receives the value and assigns it to $a$, are encoded as the processes resp. below:  
{\small 
\begin{equation*}
    \begin{aligned}
        tx.\texttt{send}(val):&\ s[p][q]\!\oplus\! \text{x}(val)& \quad 
        a = rx.\texttt{recv}():&\ s[q][p]\& \text{x}(a)& \\
    \end{aligned}
\end{equation*}
}
\noindent where the threads of $tx$ and $rx$  (i.e. $p$ and $q$) are translated as roles,  and the channel name $x$ is translated as the message label. Through this way, each thread plays one role, which satisfies the condition of progress property of the type system of Fig.~\ref{fig:typesystemfull}.  

The transmitter can be cloned to multiple copies, denoted by $tx_i, i\in I$, each of which communicates with the same receiver $rx$. The existential branching process models this case. Suppose $x_i=(tx_i, rx)$, and the threads of $tx_i$ are $p_i$ for $i\in I$, then 
{\small 
\begin{equation*}
    \begin{aligned}
        tx_i.\texttt{send}(val):&\ c[p_i][q]\oplus \text{x}_i(val)&  \quad 
        a = rx.\texttt{recv}():&\ \exists_{i \in I}c[q][p_i]\&\{\text{x}_i(a)\}& \\
    \end{aligned}
    \label{eqn:sendr}
\end{equation*}
}
If the channel has a buffer with size $n > 0$, the transmitter can send without waiting for the first $n$ times. When the buffer is full, the transmitter blocks; when the buffer is empty, the receiver blocks.  To model this case, we introduce a server thread $s_c$, which maintains a queue $q_s$ of capacity $n$ to save the messages and schedule between the transmitters and the receiver.  Its behavior is modelled by $P_0, P_1, ..., P_n$ as follows (denote $q$ by $p_N$ for some $N \notin I$):
\[
\small{
\begin{array}{lll}
     P_0 &= &\exists_{i\in I}c[s_c] [p_i] \&\{\text{x}_i(\_).c[s_c][p_i]\& \text{tx}_i(q_s             .push()).P_{1}\}_{i\in I}  \\
     P_k &= &\exists_{i\in I\cup \{N\}}c[s_c][p_i] \left\{ 
     \begin{array}{ll}
          \text{x}_i(\_).c[s_c][p_i]\& \text{tx}_i(q_s.push()).P_{k+1} \ & i \in I\\
            \text{x}_i(\_).c[s_c][p_i]\!\oplus\! \text{rx}(q_s.pop()).P_{k-1} &i = N\\
     \end{array}          
          \right.  \\ 
     &&\mbox{for} \ k\in \mathbb{N},\ 0 < k < n \\
          P_n &= &c[s_c][p_N]\& \text{x}_N(\_).c[s_c][p_N]\oplus \text{rx}(q_s.pop()).P_{n-1}  
\end{array}
}
\]
Here,  $P_i$ represents that $s_c$ has got $i$ elements in queue $q_x$.  $P_0$ can only receive messages from one of the transmitters $p_i$ and then continues as $P_1$ (it first builds connection with $p_i$ via a message labelled by $x_i$,  then receives the value from $p_i$ labelled by the corresponding transmitter name $tx_i$ and pushes it into $q_s$); For any $1<k<n$, $P_k$ can receive from transmitters $p_i$ and continue as $P_{k+1}$, or send to receiver $q$ (i.e. $p_N$) and continue as $P_{k-1}$; and $P_n$ can only send to receiver $q$. 
If the channel is asynchronous, it can be considered as a synchronous channel with an infinite buffer. 

\paragraph{\textbf{Mutex}}

For shared state concurrency, Rust provides mutex to allow only one thread to access the shared data at one time. Assume threads $t_i, i \in I$ want to access shared data $d$, 
protected by the mutex lock $m$. 
Each thread can call $\texttt{lock}$, $\texttt{try\_lock}$, $\texttt{unlock}$\footnote{Rust releases locks implicitly, but we treat $\texttt{unlock}$ explicitly for illustrating our approach. } to acquire, try to acquire and release the lock resp.  
They are encoded as follows: 
\[
\small{
\begin{array}{lll}
    m.\texttt{lock}():&m[t_i][s_m]\oplus \text{l}_i(\_).m[t_i][s_m]\oplus \text{lock}(\_).m[t_i][s_m]\& \text{ok}(d)\\
    m.\texttt{try\_lock}():&m[t_i][s_m]\oplus \text{l}_i(\_).m[t_i][s_m]\oplus \text{try\_lock}(\_).m[t_i][s_m]\& \{\text{ok}(d), \text{block}(\_)\}\\
    \texttt{unlock}(m):&m[t_i][s_m]\oplus \text{l}_i(\_).m[t_i][s_m]\oplus \text{unlock}(d)\\
\end{array}
}
\]
For each method, thread $t_i$ first sends a message to server $s_m$ to build the connection, followed by the lock/try\_lock/unlock requests resp. Message $\text{ok}$ received from $s_m$ indicates that $t_i$ acquires the lock successfully, while message $\text{block}$ means that the lock is occupied by another thread and $t_i$ fails to get it.
 
A server thread $s_m$ is defined to coordinate the access of the shared data from different threads as follows. $P$ models the initial state, expects to receive a message from $t_i, i\in I$, and for either lock or try\_lock request,  sends $\text{ok}$ to $t_i$. The process continues as $P_i$, which is available to receive messages from either $P_i$
to release the lock and go back to $P$ again, or from another $P_j$ to try to acquire the lock but it is denied and goes back to $P_i$. 

\vspace{-5mm}
\[
\small{
    \begin{array}{lll}
        P = &\exists_{i \in I} m[s_m][t_i]\&\{\text{l}_i(\_).m[s_m][t_i]\&\left\{\begin{aligned}&\text{lock}(\_).m[s_m][t_i]\oplus \text{ok}(d).P_i,&  \\&\text{try\_lock}(\_).m[s_m][t_i]\oplus \text{ok}(d).P_i \end{aligned}\right\}& \\
        P_i = &\exists_{j \in I} m[s_m][t_j]\&\{\text{l}_j(\_).P_{ij}\}, i \in I\\
        P_{ij} = &\left\{\begin{aligned}&m[s_m][t_i]\& \text{unlock}(d).P & ,i = j, i, j\in I\\&m[s_m][t_j]\&\text{try\_lock}(\_).m[s_m][t_j]\oplus \text{block}(\_).P_i &,i \neq j,i, j \in I  \end{aligned}\right.
    \end{array}}
\]
\vspace{-4mm}

\paragraph{\textbf{RwLock}}

Rust provides Read-Write Lock (abbreviated as RwLock) to allow a number of readers or at most one writer for the shared data at any point in one time. The write portion of the lock allows modification of the underlying data, which is exclusive, and the read portion allows for read-only access, which can be shared. Different from mutex,  when a reader holds the lock, other readers are still allowed to acquire the lock. In the following definitions, $t_i, m, S, s_m$ are the same as defined before.


For user thread $t_i$($i\in I$), the pair of methods $\texttt{read}$ and $\texttt{drop}$ acquire and release a guard\footnote{$ReadGuard$ is a structure used to handle the shared read access of a lock, which is created by $\texttt{read}$() and $\texttt{try\_read}$ implicitly, and is dropped when releasing. $WriteGuard$ is a similar structure for the shared write access of a lock.} corresponding to a read lock, while the pair of $\texttt{write}$ and $\texttt{drop}$ do the same for  write lock. They are modelled very similarly with mutex's lock and unlock, defined as follows. There are also corresponding $\texttt{try\_read}$ and $\texttt{try\_write}$ versions for RwLock, but we omit them here.

\[
\small{
\begin{array}{lll}
     m.\texttt{read}()&:m[t_i][s_m]\oplus \text{l}_i(\_).m[t_i][s_m]\oplus \text{read}(\_).m[t_i][s_m]\& \text{ok}(S)\\ \texttt{drop}(ReadGuard(m))&:m[t_i][s_m]\oplus \text{l}_i(\_).m[t_i][s_m]\oplus \text{readEnd}(\_)\\ 
     m.\texttt{write}()&:m[t_i][s_m]\oplus \text{l}_i(\_).m[t_i][s_m]\oplus \text{write}(\_).m[t_i][s_m]\& \text{ok}(S)\\
     \texttt{drop}(WriteGuard(m))&:m[t_i][s_m]\oplus \text{l}_i(\_).m[t_i][s_m]\oplus \text{writeEnd}(S)
\end{array}}
\]

\oomit{
\begin{equation}
\begin{aligned}
    &m.\texttt{read}():m[t_i][s_m]\oplus \text{l}_i(\_).m[t_i][s_m]\oplus \text{read}(\_).m[t_i][s_m]\& \text{ok}(S)\\
    &m.try\_read():m[t_i][s_m]\oplus \text{l}_i(\_).m[t_i][s_m]\oplus \text{try\_lock}(\_).m[t_i][s_m]\& \{\text{ok}(S), \text{block}(\_)\}\\
    &\texttt{drop}(ReadGuard(m)):m[t_i][s_m]\oplus \text{l}_i(\_).m[t_i][s_m]\oplus \text{readEnd}(\_)\\
    &m.\texttt{write}():m[t_i][s_m]\oplus \text{l}_i(\_).m[t_i][s_m]\oplus \text{write}(\_).m[t_i][s_m]\& \text{ok}(S)\\
    &m.try\_write():m[t_i][s_m]\oplus \text{l}_i(\_).m[t_i][s_m]\oplus \text{try\_write}(\_).m[t_i][s_m]\& \{\text{ok}(S), \text{block}(\_)\}\\
    &\texttt{drop}(WriteGuard(m)):m[t_i][s_m]\oplus \text{l}_i(\_).m[t_i][s_m]\oplus \text{writeEnd}(S)\\
\end{aligned}    
\end{equation}
}

For server thread $s_m$, $P$ defines the initial state, which continues as  $P_R^{\{i\}}$ or $P_W^{i}$ after thread $i$ obtains the read or write lock resp. In $P_R^A$, $A$ records the set of threads that hold the read lock, thus it allows to receive a $\text{readEnd}$ message from some thread  $x$ in $A$, followed by the removal of $x$ from $A$, and especially will go back to the initial state if $A$ only contains $x$; it also allows to receive a request from a new thread and add it to $A$. On the contrary, $P_W^i$ can only receive a $\text{writeEnd}$ message from thread $i$ holding the write lock, and goes back to the initial state.  
\oomit{
Here, $T$ is the type of the initial state, which changes to $P_R^{\{i\}}$ after process $i$ acquires the read lock, and to $P_W^{i}$ after process $i$ acquires the read lock.

$P_R^{A}$ represents the current state in which processes in set $A$ owns the lock, which changes to $P_R^{Ai}$ after receiving a message from process $i$. The type $P_R^{Ax}$ determines its own action based on the process set $A$ that currently owns the read lock and the process $x$ that currently sends the message; if $x$ is in the process set $A$ that owns the read lock, then the feasible action is to release the read lock and then remove $x$ from set $A$. Particularly, if $x$ is the only element in set $A$, then return to the initial state $P$. Otherwise, if the action is to try to acquire the read lock, then we insert $x$ to set $A$; or the action is to try to acquire the write lock, but without success, and return to $P_R^{A}$ again.

$P_W^{x}$ represents the current state in which process $x$ owns the write lock, which changes to $P_W^{xi}$ after receiving a message from process $i$. The type $P_W^{xy}$ determines its own action based on the process $x$ that currently owns the write lock and the process $y$ that currently sends the message; if $y$ is the process $x$ that owns the lock, then the feasible action is to release the lock and then return to the initial state $P$; otherwise, the action is to try to acquire the lock, but without success, and then return to $P_W^{x}$ again.
}
\[
\small{
\begin{array}{ll}
      P = &\exists_{i \in I}m[s_m][t_i]\&\{\text{l}_i(\_).m[s_m][t_i]\&\left
      \{\begin{array}{ll}
      \text{read}(\_).m[s_m][t_i]\oplus \text{ok}(S).P_R^{\{i\}},\\
       \text{write}(\_).m[s_m][t_i]\oplus \text{ok}(S).P_W^{i}     
      \end{array}\right\} \\
      P_R^A = &\exists_{i \in I} m[s_m][t_i]\&\{\text{l}_i(\_).P_R^{Ai}\}, A \subseteq I\\
        P_R^{Ax} = &\left\{\begin{array}{lll}
        &m[s_m][t_x]\& \text{readEnd}(\_).P, & A=\{x\},\quad A \subseteq I\\
        &m[s_m][t_x]\& \text{readEnd}(\_).P_R^{A\backslash \{x\}},  & \{x\} \subset A,\quad A \subseteq I\\
        &m[s_m][t_x]\& \text{read}(\_).m[s_m][t_x]\oplus \text{ok}(S).P_R^{A\cup \{i\}}, & x \notin A,\quad A \subseteq I
        \end{array} \right.\\
   P_W^{i} = & m[s_m][t_i]\& \text{writeEnd}(S).P      
\end{array}
}
\]

\paragraph{\textbf{Discussion on Rust Programs}}

With the help of the encoding of concurrency primitives, we can represent a whole Rust multi-threaded program as an MSSR process. Some details such as assignments are neglected (similar to most static analysis approaches), while concurrency primitives and control flow such as sequential composition, if-else, loop, and (recursive) functions, are supported by MSSR calculus directly. After a Rust program is translated to an MSSR process, it can be typed according to the type systems defined in this paper, to check whether it is type-safe and deadlock-free. A Rust program that passes the type checking is guaranteed to be type-safe and deadlock-free, but the inverse does not hold, i.e. programs that fail in the check might be safe in execution, mainly due to abstraction of the encoding from Rust to session calculus.


\oomit{
For the processing of conditional statements, we use merge operator $\sqcap$, whose meaning has been shown before. For loops, we use the type of $\mu t.T$ for processing and determine whether $T$ is safe. For function definition and call, we check the safe nature of the function body at function definition time according to the type rules. At function call time, we only need to determine whether the input parameters of the function are of the correct type so that we can determine whether the part calling the function is well-typed.
}

\section{Conclusion and Future Work}
\label{sec:conclusion}

This work introduced an extension of session types with existential branching type for specifying and validating more flexible interactions with multiple senders and a single receiver. This extension can represent many communication protocols, such as Rust multi-thread primitives as encoded in the paper. We established the type system for the extended session types based on projection and global types, to ensure type safety, progress and fidelity. We proposed a more general communication type system for guaranteeing progress without restrictions of existing approaches, by studying the execution orders between communication events among different participants. 

To make full use of the extended session types, 
several future works can be considered. First, the type systems presented in this paper follow a top-down approach based on end-point projections, which might be restrictive to reject some valid processes when the corresponding global types do not exist. To avoid this, we will consider to define the type systems by checking local session types directly instead, as in~\cite{scalas2019less}. Second, as an important application, we are considering to implement the translation from Rust multi-threaded programs to the extended session types for safety checking; Third, we will apply our approach to model and check more practical protocols and applications.   

\oomit{
\section{Implementation}

For the use of one end of a channel in a thread, we distinguish its occurrences in order, and the i-th use of the channel end $rx$($tx$) is noted as $rx_i$($tx_i$).
After typing, matched $tx_i$ and $rx_i$, we can check the type of both and thus determine if the communication is secure.

\begin{equation}
    \begin{aligned}
        &s[p][q]\oplus l . s'[p'][q']\oplus l'\\
        &s[q][p]\& l . s'[q'][p']\& l'
    \end{aligned}
\end{equation}

\begin{equation}
    \begin{aligned}
        &s[p][q]\oplus l . s'[p'][q']\oplus l'\\
        &s'[q'][p']\& l'.s[q][p]\& l
    \end{aligned}
\end{equation}

\begin{equation}
    \begin{aligned}
        &c.send(x); d.send(y);\\
        &d.recv();c.recv();
    \end{aligned}
\end{equation}
}

\bibliography{main}
\bibliographystyle{plain}

\end{document}